\documentclass[12pt]{article}
\usepackage{fullpage}
\usepackage{amsmath,amssymb,amsfonts,amsthm}
\usepackage{enumerate}
\usepackage{epsfig}
\usepackage{colortbl}

\usepackage{geometry}
\usepackage{tikz}
\usetikzlibrary{arrows,positioning,shadows}
\usepackage{xcolor}
\usepackage{mathrsfs}

\newif\ifnotesw\noteswtrue
   {\ifnotesw\marginpar[\hfill\(\top\)]{\(\top\)}\fi}%
   {\ifnotesw\marginpar[\hfill\(\bot\)]{\(\bot\)}\fi}

\newcommand{\mnote}[1]%
    {\ifnotesw\marginpar%
        [{\scriptsize\begin{minipage}[t]{\marginparwidth}
        \raggedleft#1%
                        \end{minipage}}]%
        {\scriptsize\begin{minipage}[t]{\marginparwidth}
        \raggedright#1%
                        \end{minipage}}%
    \fi}

\newcommand{\ignore}[1]{}
\newcommand{\etal}{{\it et al.~}}

\newtheorem{theorem}{Theorem}
\newtheorem{corollary}[theorem]{Corollary}
\newtheorem{lemma}[theorem]{Lemma}
\newtheorem{proposition}[theorem]{Proposition}

\newtheorem{fact}{Fact}

\newtheorem{definition}{Definition}

\newcommand{\iverson}[1]{\lbrack\!\lbrack #1 \rbrack\!\rbrack}

\newcommand{\ZZ}{\mathbb{Z}}
\newcommand{\RR}{\mathbb{R}}
\newcommand{\QQ}{\mathbb{Q}}
\newcommand{\GG}{\mathcal{G}}

\newcommand{\ket}[1]{\vec{#1}}
\newcommand{\bra}[1]{\vec{#1}^{\dagger}}
\newcommand{\bracket}[2]{\langle {#1} | {#2} \rangle}           
\newcommand{\tbracket}[3]{\langle {#1} | {#2} {#3} \rangle}     

\newcommand{\uket}[1]{{\mathbf{e}}_{#1}}		
\newcommand{\tbra}[1]{\uket{#1}^{T}}			
\newcommand{\buket}[1]{{\mathbf{e}}_{#1}}		
\newcommand{\vket}[1]{\vec{#1}}      			
\newcommand{\bvket}[1]{\vec{#1}}            	

\newcommand{\cart}{\mbox{ $\Box$ }}

\DeclareMathOperator{\diag}{diag}
\DeclareMathOperator{\Circ}{Circ}

\DeclareMathOperator{\Sp}{Spec}
\DeclareMathOperator{\rk}{rank}

\newcommand{\one}{\mathbf{1}}

\newcommand{\Id}{\mathbb{I}}

\newcommand{\DLp}{{L}}
\newcommand{\SLp}{{Q}}
\newcommand{\NLp}{\mathcal{L}}

\newcommand{\Line}{\ell}

\newcommand{\NA}{\mathcal{A}}
\newcommand{\NB}{\mathcal{B}}
\newcommand{\NP}{\mathcal{P}}
\newcommand{\comp}[1]{\overline{#1}}
\newcommand{\pst}{\tau}
\newcommand{\talpha}{\tilde{\alpha}}

\newcommand{\xy}{\textit{XY}}
\newcommand{\xyz}{\textit{XYZ}}

\newcommand{\Cone}[1]{\widehat{#1}}
\newcommand{\OU}{\mathcal{U}}
\newcommand{\Fan}{\Cone{P_{4}}}

\title{
Perfect State Transfer in Laplacian Quantum Walk
}
\author{
Rachael Alvir\thanks{Department of Mathematics and Statistics, Colorado Mesa University.}
\and
Sophia Dever\thanks{Department of Mathematics, University of Texas at Austin.}
\and
Benjamin Lovitz\thanks{Mathematics Department, Bates College.}
\and
James Myer\thanks{Department of Mathematics, SUNY Potsdam.}
\and
Christino Tamon\thanks{Department of Computer Science, Clarkson University.}
\and
Yan Xu\thanks{Department of Combinatorics and Optimization, University of Waterloo.}
\and
Hanmeng Zhan\footnotemark[6]
\ignore{REU 2014 Algebraic Graph Theory (draft)}
}
\date{\today}
\begin{document}
\maketitle
\bibliographystyle{plain}

\begin{abstract}
For a graph $G$ and a related symmetric matrix $M$,
the continuous-time quantum walk on $G$ relative to $M$ is defined
as the unitary matrix $U(t) = \exp(-itM)$, where $t$ varies over the reals.
Perfect state transfer occurs between vertices $u$ and $v$ at time $\tau$
if the $(u,v)$-entry of $U(\tau)$ has unit magnitude. 
This paper studies quantum walks relative to graph Laplacians. 
Some main observations include the following closure properties for perfect state transfer:
\begin{itemize}
\item
If a $n$-vertex graph has perfect state transfer at time $\tau$ relative to the Laplacian, 
then so does its complement if $n\tau \in 2\pi\ZZ$.
As a corollary, the double cone 
over any $m$-vertex graph has perfect state transfer relative to the Laplacian 
if and only if $m \equiv 2\pmod{4}$. 
This was previously known for a double cone over a clique
(S. Bose, A. Casaccino, S. Mancini, S. Severini, Int. J. Quant. Inf., {\bf 7}:11, 2009).

\item 
If a graph $G$ has perfect state transfer at time $\tau$ relative to the normalized Laplacian, 
then so does the weak product $G \times H$ if 
for any normalized Laplacian eigenvalues $\lambda$ of $G$ and $\mu$ of $H$,
we have $\mu(\lambda-1)\tau \in 2\pi\ZZ$.
As a corollary, a weak product of $P_{3}$ with an even clique or an odd cube
has perfect state transfer relative to the normalized Laplacian.
It was known earlier that a weak product of a circulant with odd integer eigenvalues 
and an even cube or a Cartesian power of $P_{3}$ has perfect state transfer
relative to the adjacency matrix.
\end{itemize}
As for negative results, no path with four vertices or more has
antipodal perfect state transfer relative to the normalized Laplacian.
This almost matches the state of affairs under the adjacency matrix 
(C. Godsil, Discrete Math., {\bf 312}:1, 2011).

\vspace{.1in}
\par\noindent{\em Keywords}: 
quantum walk, perfect state transfer.
Laplacian (combinatorial, signless, normalized), 
equitable and almost-equitable partitions, join, weak product, line graph.
\end{abstract}


\section{Introduction}

Given a graph $G=(V,E)$, we may associate a matrix $M$ with $G$. 
For example, in graph theory, common choices for $M$ include the adjacency matrix $A$ 
and the Laplacian $ D-A$, where $D$ is the diagonal degree matrix of $G$.
On the other hand, in probability theory, natural choices for $M$ include 
a simple random walk matrix $P = AD^{-1}$ and a lazy random walk matrix 
$W = \frac{1}{2}\Id + \frac{1}{2}P$.
If $M$ is Hermitian, then we may define
a continuous-time quantum walk on $G$ relative to $M$ as the time-dependent unitary matrix
\begin{equation} \label{eqn:def-qwalk}
U_{G}(t) = \exp(-itM),
\end{equation} 
where $t \in \RR$. 
This definition is motivated by Schr\"{o}dinger's equation where $M$ is viewed as 
the Hamiltonian of the underlying system.
Continuous-time quantum walk on graphs is a useful method for designing efficient 
quantum algorithms (see Childs \etal \cite{ccdfgs03} and Farhi \etal \cite{fgg08}) 
and is a universal model for quantum computation (see Childs \cite{childs09}).

In an early seminal work, Farhi and Gutmann \cite{fg98} used the {\em infinitesimal generator} 
matrix to define their quantum walk. The latter matrix is a weighted Laplacian matrix used
commonly to define a continuous-time random walk (see Grimmett and Stirzaker \cite{gs80}).
This Laplacian matrix provides arguably the most natural connection between the continuous-time 
classical random walk and its quantum counterpart. 
As pointed out by Bose \etal \cite{bcms09}, from a physics viewpoint, 
the quantum walks relative to the adjacency and Laplacian matrices are intimately related 
to quantum spin chains in the $\xy$ and $\xyz$ interaction models, respectively. 
The $\xyz$ interaction model is also known as the isotropic Heisenberg model.

The literature on graph Laplacians is vast and has a strong focus on the following
three different Laplacians.
The aforementioned {\em standard} (or combinatorial) Laplacian $D-A$ is closely related 
to Laplace's heat equation and has beautiful algorithmic applications 
(see Spielman \cite{spielman-survey12}).
The {\em signless} Laplacian $D+A$ of a graph $G$ shares a strong spectral correspondence with
the line graph $\Line(G)$ through the incidence matrix of $G$. 
The {\em normalized} Laplacian $\NLp = D^{-1/2}(D - A)D^{-1/2}$ has an interesting connection 
to the Heat Kernel random walk which is defined as $e^{-t\NLp}$
(see Chung's monograph \cite{chung}). 
Since $\NLp = \Id - D^{-1/2}PD^{1/2}$, where $P$ is the simple random walk matrix,
the normalized Laplacian is similar to $\Id - P$. 
But even though $e^{-it\NLp}$ is a well-defined quantum walk, 
the ``quantum walk'' $e^{-it(\Id-P)}$ is illegal since $P$ might not be symmetric.
The latter is related to the Heat Kernel random walk via an imaginary time 
transformation\footnote{This is apparently a common technique in statistical and quantum physics. 
See \cite{ik10,gvz13} for an application of this method to continuous-time random and quantum walks.}
$t \leftrightarrow it$.

A quantum walk on a graph $G$ relative to a matrix $M$ has {\em perfect state transfer} 
between vertices $u$ and $v$ at time $\tau$ if the $(u,v)$-entry of the unitary matrix $U(\tau)$
has unit magnitude; that is:
\begin{equation} \label{eqn:def-pst}
|\tbracket{\buket{v}}{e^{-i\tau M}}{\buket{u}}| = 1.
\end{equation}
The notion of state transfer was introduced by Bose \cite{bose03} in the context of 
information transfer in quantum spin chains. 
In his work, Bose considered perfect state transfer in the $\xyz$ model. 
This notion was further studied by Christandl \etal \cite{cdel04,cddekl05} 
for paths and hypercubes in the $\xy$ (adjacency matrix) model. They observed that the $n$-cube 
has antipodal perfect state transfer at time $\pi/2$, for {\em any} $n$ (which is counter-intuitive 
since the diameter of the $n$-cube increases with $n$). In contrast, relative to the normalized Laplacian, 
the antipodal perfect state transfer time of the $n$-cube is $n\pi/2$. 

Our main goal is to understand how these graph Laplacians affect state transfer on graphs
and how they compare with the adjacency matrix model. 
For regular graphs, the quantum walks relative to the adjacency and Laplacian matrices are equivalent
(up to irrelevant phase factors, time dilations, and time reversal). 
On bipartite graphs, the quantum walks relative to the standard and signless Laplacians are equivalent. 
So, our primary focus will be on irregular and/or nonbipartite graphs. 
We describe some of our results in what follows.

\begin{figure}[t]
\begin{center}
\begin{tabular}{|c|c|c|c|} \hline
\rowcolor[gray]{0.9}
Graph family                    & PST time  		& Laplacian          	& Source        \\  \hline
$Q_{n}$							& $\pi/2$  			& standard/signless		& Christandl \etal \cite{cddekl05} \\
$Q_{n}$							& $n\pi/2$  		& normalized			& Moore and Russell \cite{mr02} \\ \hline
$\comp{K_{2}}+\GG_{4n-2}$		& $\pi/2$   		& standard              & this work \\
$\comp{K_{2}}+\GG_{2n,n-1}$		& $\pi/2\sqrt{n}$	& signless              & this work \\ 
$P_{3} \times \{K_{2n},Q_{2n-1}\}$     		
								& $(2n-1)\pi$		& normalized   			& this work \\ \hline
$P_{n \ge 3}$, $T_{n \ge 3}$	& $\infty$    		& standard/signless     & Godsil \cite{godsil-book}, 
																				Coutinho and Liu \cite{cl14} \\
$P_{n \ge 4}$					& $\infty$ 			& normalized			& this work \\ \hline
\end{tabular}
\caption{
{\em Summary of some results on Laplacian perfect state transfer (PST)}:
a PST time of $\infty$ denotes {\em no} perfect state transfer;
$n \ge 1$ is a positive integer;
$\GG_{n}$ denotes any family of $n$-vertex graphs; 
$\GG_{n,k}$ denotes any family of $(n,k)$-regular graphs; 
$P_{n}$, $T_{n}$, $K_{n}$ denote a path, tree, and complete graph
on $n$ vertices, respectively;
$Q_{n}$ is the $n$-dimensional hypercube.
}
\label{fig:results}
\end{center}
\hrule
\end{figure}

For the standard Laplacian, we show that perfect state transfer is closed 
under complementation with some mild assumptions. As a corollary, we characterize
perfect state transfer on double cones:
$\comp{K_{2}} + H$ has perfect state transfer relative to the standard Laplacian
if and only if $|V(H)| \equiv 2\pmod{4}$. This generalizes a result of Bose \etal
\cite{bcms09} where $H$ is the complete graph. 
We also compare this to the $\xy$ model where the double cone has perfect state transfer
if $H$ is a $(n,k)$-regular graph, provided both $k$ and $\sqrt{k^{2}+8n}$ are integers 
divisible by four and the largest powers of two which divide them are distinct
(see Angeles-Canul \etal \cite{anoprt10}). 
Thus, perfect state transfer on double cones is less complicated in the $\xyz$ model 
compared to the $\xy$ model.
In contrast, we also find double cones with perfect state transfer 
relative to the signless Laplacian, but not the standard Laplacian.
In particular, we show $\comp{K_{2}} + H$ has perfect state transfer relative to the
signless Laplacian if $|V(H)|$ is even and it is densely regular
(more precisely, regular with degree $\frac{1}{2}|V(H)|-1$).
Most of our proofs employ the machinery of quotient graphs under equitable
and almost-equitable partitions.

By exploiting the spectral connection between the signless Laplacian and line graphs, 
we show there is no perfect state transfer on a family of odd unicyclic graphs relative 
to the signless Laplacian. The latter family of graphs is obtained by attaching two pendant 
paths to a three-cycle. Our proof uses the idea of controllable subsets in graphs 
(see Godsil \cite{godsil-ac12}).
Using the same technique, we can also show there is no perfect state transfer on paths 
with five or more vertices relative to the signless (and standard, by switching equivalence) Laplacian.
But, a better result (with optimal proof) is known to Godsil who showed that paths 
on at least three vertices have no perfect state transfer under the standard Laplacian.
Recently, Coutinho and Liu \cite{cl14} improved this considerably and showed there is 
no perfect state transfer on trees with at least three vertices relative to the standard Laplacian.

So, we know $P_{3}$ has no perfect state transfer relative to the standard/signless Laplacians
since it is a double cone over a single vertex (also, of course, from the results of Godsil, 
Coutinho and Liu mentioned above).
Interestingly, $P_{3}$ has perfect state transfer under the normalized Laplacian
at time $\pi$ (as opposed to $\pi/\sqrt{2}$ in the $\xy$ model).
We use this to show that a weak product of $P_{3}$ with either an even clique or an odd cube
has perfect state transfer relative to the normalized Laplacian.
This is a consequence of another closure property:
if $G$ has perfect state transfer under the normalized Laplacian at time $\tau$,
then so does the weak product $G \times H$ provided that 
for any normalized Laplacian eigenvalues $\lambda$ of $G$ and $\mu$ of $H$,
$\mu(\lambda-1)\tau$ is an integer multiple of $2\pi$.
In comparison, relative to the adjacency matrix, it is known that
a weak product of a circulant which has odd integer eigenvalues with
either $Q_{2n}$ or $P_{3}^{\Box n}$, where $n$ is a positive integer,
has perfect state transfer
(see Ge \etal \cite{ggpt11}). 

Finally, we show that no path on four or more vertices has antipodal perfect state transfer
relative to the normalized Laplacian. The proof is based on a reduction to even cycles
in the adjacency matrix model. This almost matches the strong negative result for paths 
in the $\xy$ model, where there is no perfect state transfer between any pair of vertices 
(see Godsil \cite{godsil-dm11}).

We summarize some of the known results on Laplacian state transfer along with our contributions 
in Figure \ref{fig:results}. A survey on state transfer from a graph-theoretic perspective 
is given by Godsil \cite{godsil-dm11}.


\section{Preliminaries}

For a logical statement $S$, we use the Iversonian bracket $\iverson{S}$ to denote
$1$ if $S$ is true, and $0$ otherwise (see \cite{gkp94}).
The $n$-dimensional all-one vector is denoted $\one_{n}$. 
The identity matrix of order $n$ is denoted $\Id_{n}$.
The $m \times n$ all-one matrix is denoted $J_{m,n}$ or simply $J_{n}$ whenever $m=n$.
We omit dimensions if the context is clear.
For a matrix $A$, $A^{T}$ and $A^{\dagger}$ denote its transpose and Hermitian transpose, respectively.
The inner product of vectors $\vket{u}$ and $\vket{v}$ is denoted $\bracket{\vket{u}}{\vket{v}}$.
Given an index $u$, let $\uket{u}$ denote the unit vector that is $1$ at position $u$ and zero elsewhere.
We often consider the inner product $\tbracket{x}{A}{y}$, and in the form of
$\tbracket{\buket{u}}{A}{\buket{v}}$, it is simply the $(u,v)$-entry of $A$.

For two sets $A,B$ of numbers, we denote 
their sum as $A + B = \{a+b : a \in A, b \in B\}$,
their product as $AB = \{ab : a \in A, b \in B\}$,
and the scalar product of $A$ with a constant $c$ as 
$cA = \{ca : a \in A\}$.

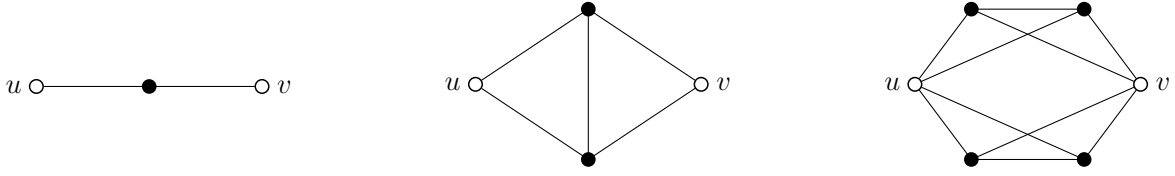
\begin{figure}[t]
\begin{center}
\begin{tikzpicture}
%
\node at (0,0)[scale=0.5]{};
\foreach \y in {+1} {	
	\draw (-1.5,\y)--(0,\y);
	\draw (+1.5,\y)--(0,\y);
	\foreach \x in {-1.5,1.5} {
		\node at (\x,\y)[circle, fill=white][scale=0.5]{};
		\draw[line width=0.2mm] (\x, \y) circle (0.09cm);
	}
	\node at (0,\y)[circle, fill=black][scale=0.5]{};

	\node at (-1.8,\y)[scale=0.9]{$u$};
	\node at (+1.8,\y)[scale=0.9]{$v$};
}
\end{tikzpicture}
\quad \quad \quad \quad
\begin{tikzpicture}
%
\foreach \y in {-1,1} {
    \draw (-1.5,0)--(0,\y);
    \draw (+1.5,0)--(0,\y);
    \node at (0,\y)[circle, fill=black][scale=0.5]{};
}
\node at (-1.5,0)[circle, fill=white][scale=0.5]{};
\draw[line width=0.2mm] (-1.5,0) circle (0.09cm);
\node at (+1.5,0)[circle, fill=white][scale=0.5]{};
\draw[line width=0.2mm] (+1.5,0) circle (0.09cm);

\draw (0,-1)--(0,+1);

\node at (-1.8,0)[scale=0.9]{$u$};
\node at (+1.8,0)[scale=0.9]{$v$};
\end{tikzpicture}
\quad \quad \quad \quad
\begin{tikzpicture}
%
\foreach \x in {-0.75,0.75} {
	\foreach \y in {-1,1} {
		\node at (\x,\y)[circle, fill=black][scale=0.5]{};
	}
}
\foreach \y in {-1,+1} {
	\draw (-0.75,\y)--(+0.75,\y);
}

\foreach \z in {-1.5,+1.5} {
	\foreach \x in {-0.75,+0.75} {
		\foreach \y in {-1,+1} {
			\draw (\z,0)--(\x,\y);
		}
	}
}

\node at (-1.5,0)[circle, fill=white][scale=0.5]{};
\draw[line width=0.2mm] (-1.5,0) circle (0.09cm);
\node at (+1.5,0)[circle, fill=white][scale=0.5]{};
\draw[line width=0.2mm] (+1.5,0) circle (0.09cm);
\node at (-1.8,0)[scale=0.9]{$u$};
\node at (+1.8,0)[scale=0.9]{$v$};
\end{tikzpicture}
\caption{
Small examples of graphs with Laplacian perfect state transfer (between vertices marked white):
(i) $P_{3} = \comp{K_{2}} + K_{1}$ has perfect state transfer at time $\pi$ relative to the normalized Laplacian 
	(but not the standard/signless Laplacian);
(ii) $\comp{K_{2}} + K_{2}$ has perfect state transfer at time $\pi/2$ relative to the standard/signless Laplacians;
(iii) $\comp{K_{2}} + 2K_{2}$ has perfect state transfer at time $\pi/\sqrt{8}$ relative to the signless Laplacian
	(but not the standard Laplacian).
}
\label{figure:lap-pst}
\end{center}
\hrule
\end{figure}

Let $G=(V,E)$ be a graph that is simple, undirected, and (mostly) connected.
Two vertices $u$ and $v$ are adjacent, or $u \sim v$, if $(u,v) \in E$.
The degree of a vertex $u \in V$, which we denote $\deg(u)$, is the number of vertices adjacent 
to $u$; that is, $\deg(u) = \sum_{v \in V} \iverson{u \sim v}$.
A graph $G$ is called $(n,k)$-regular if $G$ has $n$ vertices and each vertex has degree $k$.
As is customary, we let $P_{n}$ and $K_{n}$ denote a path and a complete graph on $n$ vertices, 
respectively, and $Q_{n}$ denote the $n$-dimensional hypercube.

\smallskip

For a graph $G=(V,E)$, its adjacency matrix $A$ is defined as 
$A_{u,v} = \iverson{(u,v) \in E}$
and its diagonal degree matrix $D$ is defined as 
$D_{u,v} = \iverson{u=v} \deg(u)$.
We focus on the following graph Laplacians.
The {\em standard} Laplacian is given by $\DLp = D - A$,
the {\em signless} Laplacian by $\SLp = D + A$,
and
the {\em normalized} Laplacian is $\NLp = \Id - D^{-1/2}AD^{-1/2}$.
For a matrix $M$ related to a graph $G$,
the $M$-{\em spectrum} of $G$, denoted $\Sp_{M}(G)$, is the set of eigenvalues of $M(G)$.

\smallskip

The {\em complement} of a graph $G$, denoted $\comp{G}$, is a graph whose vertex set is $V(G)$ with 
edge set $\{(u,v) : (u,v) \not\in E(G), u \neq v\}$.
For two graphs $G$ and $H$, their {\em disjoint union} $G \cup H$ is a graph whose vertex set is
$V(G) \cup V(H)$ and edge set is $E(G) \cup E(H)$, respectively. 
Here, we assume $V(G)$ and $V(H)$ are disjoint sets.
The {\em join} of $G$ and $H$, denoted $G + H$, is defined as 
$G + H = \overline{\overline{G} \cup \overline{H}}$.
We also consider products of $G$ and $H$ where the vertex set is $V(G) \times V(H)$
and the edge set is defined by an adjacency rule on the pairs $(g_{1},h_{1})$ and $(g_{2},h_{2})$:
\begin{itemize}
\item {\em weak product} $G \times H$:
	$(g_{1},h_{1}) \sim (g_{2},h_{2})$ if $g_{1} \sim g_{2}$ and $h_{1} \sim h_{2}$.
	The adjacency matrix is given by $A(G \times H) = A(G) \otimes A(H)$.

\item {\em Cartesian product} $G \cart H$:
	$(g_{1},h_{1}) \sim (g_{2},h_{2})$ if
	$g_{1} \sim g_{2}$ and $h_{1} = h_{2}$, or $g_{1} = g_{2}$ and $h_{1} \sim h_{2}$.
	The adjacency matrix is given by 
	$A(G \cart H) = A(G) \otimes \Id_{H} + \Id_{G} \otimes A(H)$.
\end{itemize}
The {\em line graph} of a graph $G$, denoted\footnote{We follow a convention used by Mike Newman \cite{newman-thesis}.} 
$\Line(G)$, is a graph whose vertex set is $E(G)$ where two edges are adjacent in $\Line(G)$
if they share a common vertex;
that is, $E(\Line(G)) = \{(e_{1},e_{2}) : e_{1},e_{2} \in E(G), |e_{1} \cap e_{2}| = 1\}$.

\medskip

A vertex partition $\pi$ of a graph $G=(V,E)$ given by $V = V_{1} \cup \ldots \cup V_{m}$
is called {\em equitable} if there are constants $d_{j,k}$, for $1 \le j,k \le m$, so that 
\begin{equation} \label{eqn:equitable}
(\forall j,k \in \{1,\ldots,m\})(\forall u \in V_{j}) \
|N(u) \cap V_{k}| = d_{j,k}.
\end{equation}
If condition (\ref{eqn:equitable}) is only required to hold for $j \neq k$, 
the partition is called {\em almost equitable}.

The (normalized) partition matrix $\NP$ of $\pi$ is given by
\begin{equation}
\NP = \sum_{u \in V, k \in [m]} \frac{\iverson{u \in V_{k}}}{\sqrt{|V_{k}|}} \uket{u}\tbra{k}
\end{equation}
We state the following well-known properties of equitable partitions.

\begin{fact} \label{fact:equitable-adjacency} (Godsil \cite{godsil-dm11}) \\
Let $G=(V,E)$ be a graph with an equitable partition $\pi$ given by $V = \bigcup_{k=1}^{m} V_{k}$
with constants $d_{j,k}$, for $j,k=1,\ldots,m$. Let $\NP$ be the (normalized) partition matrix of $\pi$,
where $\NP^{T}\NP = \Id_{m}$. Then:
\begin{enumerate}
\item $\NP\NP^{T} = \diag(\{J_{|V_{k}|} : k \in [m]\})$ which commutes with $A(G)$.
\item $A(G)\NP = \NP B$, where $B$ is a $m \times m$ matrix defined as
\begin{equation}
B_{j,k} = \sqrt{d_{j,k}d_{k,j}},
\end{equation}
where $j,k = 1,\ldots,m$.
\end{enumerate}
Thus, $B = A(G/\pi)$ is the adjacency matrix of the quotient graph $G/\pi$. 
\end{fact}

We will need the following lemma which relates perfect state transfer in quantum walks
on a graph and on its quotient under an equitable partition.

\begin{lemma} \label{lemma:lifting} (Bachman \etal \cite{bfflott12}) \\
Let $G=(V,E)$ be a graph with equitable partition $\pi$.
Suppose $u,v \in V(G)$ belong to singleton partitions under $\pi$.
Then,
\begin{equation}
\tbracket{\buket{u}}{e^{-itA(G)}}{\buket{v}}
=
\tbracket{\buket{\pi(u)}}{e^{-itA(G/\pi)}}{\buket{\pi(v)}}.
\end{equation}
Therefore, 
perfect state transfer occurs between $u$ and $v$ in $G$ 
if and only if
it occurs between $\pi(u)$ and $\pi(v)$ in $G/\pi$.
\end{lemma}

Further background on algebraic graph theory can be found in 
Godsil and Royle \cite{godsil-royle01}.


\section{Basic observations}

In this section, we state some basic facts about Laplacian quantum walk on graphs. 

\begin{definition} (Equivalence under quantum walk) \\
Given a graph $G$ and two matrices $M_{1}(G)$ and $M_{2}(G)$ associated with $G$, 
the quantum walks based on $M_{1}(G)$ and $M_{2}(G)$ are {\em equivalent}
if for every time $t \in \RR$, we have
\begin{equation}
|\tbracket{\buket{u}}{e^{-itM_{1}(G)}}{\buket{v}}| 
= 
|\tbracket{\buket{u}}{e^{-i(\alpha t)M_{2}(G)}}{\buket{v}}|,
\end{equation}
for each $u,v \in V(G)$ and for some $\alpha \in \RR$.
\end{definition}

Here, we consider two quantum walks equivalent if their entry-wise complex magnitudes are the same
at all times. The global phase factors (of the form $e^{i\theta}$ for some real $\theta$) may be
safely ignored since they are undetectable by quantum measurements.

\subsection{Regular graphs}

\begin{fact} \label{fact:regular}
For any regular graph $G$, 
the quantum walks based on the adjacency matrix, 
the standard and signless Laplacians, 
and the normalized Laplacian are all equivalent.
\end{fact}
\begin{proof}
Let $G$ be a $k$-regular graph. 
The standard, signless and normalized Laplacians of $G$, respectively, are given by 
$\DLp(G) = k\Id - A(G)$,
$\SLp(G) = k\Id + A(G)$,
and
$\NLp(G) = \Id - \frac{1}{k}A(G)$.
Therefore, their quantum walks are defined as 
\begin{eqnarray}
\exp(-it\DLp(G)) & = & e^{-i(kt)}\exp(+itA(G)) \\
\exp(-it\SLp(G)) & = & e^{-i(kt)}\exp(-itA(G)) \\
\exp(-it\NLp(G)) & = & e^{-it}\exp(i(t/k)A(G)),
\end{eqnarray}
which are all equivalent to $e^{-itA(G)}$ up to phase factors, time reversal and time dilations.
\end{proof}

\subsection{Bipartite graphs}

\begin{fact} \label{fact:bipartite}
For any bipartite graph $G$, 
the quantum walks based on the standard and signless Laplacians are equivalent.
\end{fact}
\begin{proof}
If $G$ is bipartite, then $-A(G) = DA(G)D^{-1}$ for some nonsingular diagonal matrix
$D$ with $\pm 1$ entries along its diagonal
(see Godsil and Royle \cite{godsil-royle01}, for example).
This implies that $\SLp(G) = D\DLp(G)D^{-1}$ and, moreover, $e^{-it\SLp(G)} = De^{-it\DLp(G)}D^{-1}$.
\end{proof}

\subsection{Cartesian products}

To construct infinite families of graphs with perfect state transfer in the $\xy$ model,
the Cartesian product is a useful closure operator. The seminal works of Christandl \etal
\cite{cdel04,cddekl05} showed that both $K_{2}^{\Box n}$ and $P_{3}^{\Box n}$ have
perfect state transfer (since each of $K_{2}$ and $P_{3}$ have such property).
We state a similar observation for the standard/signless Laplacians.

\begin{fact} \label{fact:laplacian-cartesian-product}
Let $M$ denote the standard or signless Laplacian.
Suppose $G$ has perfect state transfer at time $t$ between $g_{1}$ and $g_{2}$ relative to $M$.
Suppose $H$ has perfect state transfer at time $t$ between $h_{1}$ and $h_{2}$ relative to $M$.
Then, $G \cart H$ has perfect state transfer at time $t$ between $(g_{1},h_{1})$ and $(g_{2},h_{2})$ 
relative to $M$.
\end{fact}
\begin{proof}
Note that $D(G \cart H) = D(G) \otimes \Id_{H} + \Id_{G} \otimes D(H)$.
Thus, $M(G \cart H) = M(G) \otimes \Id_{H} + \Id_{G} \otimes M(H)$.
This shows that
\begin{equation}
\exp(-itM(G \cart H))
=
e^{-itM(G)} \otimes e^{-itM(H)},
\end{equation}
which implies the claim.
\end{proof}

Later, we apply Fact \ref{fact:laplacian-cartesian-product} to some examples of graphs 
with Laplacian perfect state transfer. 
For example, $Q_{n} \cart (\comp{K_{2}} + G)$ has perfect state transfer at time $\pi/2$ 
relative to the standard Laplacian, for any $n$-cube $Q_{n}$ and any graph $G$ with 
$|V(G)| \equiv 2\pmod{4}$ (see Corollary \ref{cor:laplacian-double-cone}). 
See Figure \ref{fig:product}.

\subsection{Three-vertex path}

\begin{figure}[t]
\begin{center}
\begin{tikzpicture}
\node at (0,-0.25){};
{
\draw(-1.5, 0)--(+1.5, 0);

\draw (0, 0) .. controls (-0.5, 0.5) and (-0.5, 0.75) .. (0, 0.75) .. controls (+0.5, 0.75) and (+0.5, 0.5) .. (0, 0);

\node at (0, 1)[scale=0.8]{$\alpha$};

\node at (0, 0)[circle, fill=black][scale=0.5]{};

\foreach \x in {-1.5, 1.5}
{
    \node at (\x, 0)[circle, fill=white][scale=0.5]{};
    \draw[line width=0.2mm] (\x, 0) circle (0.09cm);
}

\node at (-1, 0.25)[scale=0.8]{$1$};
\node at (+1, 0.25)[scale=0.8]{$1$};
}
\end{tikzpicture}
\caption{The weighted graph $P_{3}(\alpha)$:
antipodal perfect state transfer occurs at time $\tau$ relative to the adjacency matrix
if and only if $\cos(\tau\alpha/2)\cos(\tau\Delta) = -1$,
where $\Delta^{2} = (\alpha/2)^{2} + 2$.
}
\end{center}
\hrule
\end{figure}
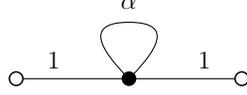

We show that $P_{3}$ has perfect state transfer relative to the normalized Laplacian.

\begin{fact} \label{fact:normalized-p3}
$P_{3}$ has antipodal perfect state transfer relative to the normalized Laplacian at time $\pi$.
\end{fact}
\begin{proof}
Note that $\NLp(P_{3}) = \Id - \frac{1}{\sqrt{2}}A(P_{3})$ with eigenvalues $\lambda = 0,1,2$.
Therefore, we have
\begin{equation}
e^{-it\NLp(P_{3})} 
= \exp\left(-it\left[\Id - \frac{1}{\sqrt{2}}A(P_{3})\right]\right) 
= e^{-it} e^{i(t/\sqrt{2})A(P_{3})}.
\end{equation}
Since $A(P_{3})$ has antipodal perfect state transfer at time $\pi/\sqrt{2}$ (see Godsil \cite{godsil-dm11}),
$P_{3}$ has antipodal perfect state transfer at time $\pi$ relative to the normalized Laplacian.
\end{proof}

\medskip

In what follows, we consider a path on three vertices with a weighted self-loop on the middle
vertex. We describe a necessary and sufficient condition on the weight of the self-loop that yields
antipodal perfect state transfer. This simple graph will be useful later when we analyze 
the double cone $\comp{K_{2}} + G$ on Laplacians. 

\begin{fact} \label{fact:p3-optimal}
For a real number $\alpha \in \RR$, let $\GG(\alpha)$ be a graph on the vertex set $\{0,1,2\}$
with the following adjacency matrix:
\begin{equation}
A(\GG(\alpha)) =
\begin{bmatrix}
0 & 1 & 0 \\
1 & \alpha & 1 \\
0 & 1 & 0
\end{bmatrix}
\end{equation}
Then, $G(\alpha)$ has antipodal perfect state transfer relative to the adjacency matrix 
at time $\tau$ if and only if
\begin{equation} \label{eqn:p3-iff}
e^{-i\tau\alpha/2}\cos(\Delta \tau) = -1,
\end{equation}
where $\Delta = \sqrt{(\alpha/2)^{2} + 2}$.
\end{fact}
\begin{proof}
Let $\talpha = \alpha/2$ and $\Delta = \sqrt{\talpha^{2} + 2}$.
The eigenvalues of $A(\GG(\alpha))$ are
$\lambda_{0} = 0$ and 
$\lambda_{\pm} = \talpha \pm \Delta$.
with the following corresponding eigenvectors:
\begin{equation}
\vket{z_{0}} = 
\frac{1}{\sqrt{2}}
\begin{bmatrix}
1 \\ 0 \\ -1
\end{bmatrix},
\ \hspace{.3in}
\vket{z_{\pm}} = 
\frac{1}{\sqrt{2\Delta(\Delta \pm \talpha)}}
\begin{bmatrix}
1 \\ \lambda_{\pm} \\ 1
\end{bmatrix}
\end{equation}
where we have used $2 + \lambda_{\pm}^{2} = 2\Delta(\Delta \pm \talpha)$.

Thus, the antipodal fidelity of the quantum walk $e^{-itA(\GG(\alpha))}$ is given by
\begin{eqnarray}
\tbracket{\buket{2}}{e^{-itA(\GG(\alpha))}}{\buket{0}}
	& = & -\frac{1}{2} + \frac{1}{2\Delta}\sum_{\pm} \frac{e^{-it\lambda_{\pm}}}{(\Delta \pm \talpha)} \\
\ignore{
	& = & -\frac{1}{2} + \frac{e^{-it\talpha}}{2\Delta}
		\left[
		\frac{(\Delta - \talpha)e^{-it\Delta} + (\Delta + \talpha)e^{it\Delta}}{\Delta^{2} - \talpha^{2}}
		\right],
		\mbox{ since $\lambda_{\pm} = \talpha \pm \Delta$ } \\
	& = & -\frac{1}{2} + \frac{e^{-it\talpha}}{4\Delta}
		\left[
		(\Delta - \talpha)e^{-it\Delta} + (\Delta + \talpha)e^{it\Delta}
		\right],
		\mbox{ since $\Delta^{2} = \talpha^{2} + 2$} \\
}
	\label{eqn:fidelity-p3}
	& = & -\frac{1}{2} + \frac{e^{-it\talpha}}{2}
		\left[\cos(\Delta t) + i\left(\frac{\talpha}{\Delta}\right)\sin(\Delta t)\right] 
\end{eqnarray}
Let $g(\alpha,t) = \cos(\Delta t) + i(\talpha/\Delta)\sin(\Delta t)$.
Note that $|g(\alpha,t)| = 1$ if and only if $\Delta t \in \ZZ\pi$,
since $\talpha/\Delta < 1$.

We show that there is a time $\tau$ so 
$|\tbracket{\buket{2}}{e^{-i\tau A(\GG(\alpha))}}{\buket{0}}| = 1$ 
if and only if $e^{-i\tau\talpha}\cos(\Delta \tau) = -1$.
By inspecting (\ref{eqn:fidelity-p3}), the condition is clearly sufficient.
To show it is necessary, suppose there are $\tau,\theta \in \RR$ so that
$\tbracket{\buket{2}}{e^{-i\tau A(\GG(\alpha))}}{\buket{0}} = e^{i\theta}$.
From (\ref{eqn:fidelity-p3}), we have
\begin{equation} \label{eqn:middle-eqn}
e^{i\theta} = -\frac{1}{2} + \frac{1}{2} e^{-i\tau\talpha}g(\alpha,\tau).
\end{equation}
This implies that $|g(\alpha,\tau)| = 1$ (by taking the complex conjugate and multiplying).
Note that if a convex combination of numbers of the form $e^{i\beta_{k}}$ lies on the complex unit circle, 
then all $\beta_{j}$ are congruent modulo $2\pi$.
Hence, $e^{-i\tau\talpha}\cos(\Delta \tau) = -1$.
\end{proof}


\section{Standard Laplacian}

\subsection{Complements}

We show that perfect state transfer relative to the Laplacian is closed under complementation. 
Relative to the adjacency matrix, this only holds for regular graphs.

\begin{theorem} \label{thm:laplacian-complement}
If $G$ is a graph with perfect state transfer between vertices $u$ and $v$ at time $t$ 
relative to the standard Laplacian, where
\begin{equation}
|V(G)|t \in 2\pi\ZZ,
\end{equation}
then $\comp{G}$ has perfect state transfer between vertices $u$ and $v$ at time $t$ 
relative to the standard Laplacian.
\end{theorem}
\begin{proof}
Let $G$ be a graph on $n$ vertices. 
The standard Laplacian of $\comp{G}$ is given by
\begin{equation}
\DLp(\comp{G}) = [(n-1)\Id - D(G)] - [J - \Id - A(G)] = n\Id - J - \DLp(G).
\end{equation}
Since $\DLp(G)$ commutes with $J$, we get
\begin{equation}
e^{-it\DLp(\comp{G})} 
= e^{-int} e^{itJ} e^{it\DLp(G)}.
\end{equation}
By the spectral theorem, 
$e^{itJ} = e^{int}J/n + \Id - J/n$, which implies
\begin{equation}
e^{-it\DLp(\comp{G})} 
= \textstyle e^{-int} \left[ (e^{int}-1)\frac{1}{n}J + \Id \right] e^{it\DLp(G)}.
\end{equation}
Thus, if $nt \in 2\pi\ZZ$, we obtain $e^{-it\DLp(\comp{G})} = e^{it\DLp(G)}$.
\end{proof}

\smallskip

We show applications of Theorem \ref{thm:laplacian-complement} to perfect state transfer on
graph joins and on double cones relative to the standard Laplacian.

\begin{corollary} \label{cor:laplacian-join}
Let $\comp{G}$ be a graph which has perfect state transfer between vertices $u$ and $v$ at time $t$
relative to the standard Laplacian. 
For any graph $H$, the join $G + H$ has perfect state transfer between vertices $u$ and $v$
at time $t$ relative to the standard Laplacian provided
\begin{equation}
t(|V(G)| + |V(H)|) \in 2\pi\ZZ.
\end{equation}
\end{corollary}
\begin{proof}
We note that $G + H = \comp{\comp{G} \cup \comp{H}}$ and apply Theorem \ref{thm:laplacian-complement}.
\end{proof}

\smallskip

\begin{corollary} \label{cor:laplacian-double-cone}
The join $\comp{K_{2}} + H$ has perfect state transfer at time $\pi/2$ between the vertices of
$\comp{K_{2}}$ relative to the standard Laplacian 
if $|V(H)| \equiv 2\pmod{4}$.
\end{corollary}
\begin{proof}
We apply Corollary \ref{cor:laplacian-join} with $\comp{G} = K_{2}$ which has perfect state transfer
at time $t = \pi/2$. Thus, $\comp{K_{2}}+H$ has perfect state transfer if
$(2 + |V(H)|)\pi/2 \in 2\pi\ZZ$, which proves the claim.
\end{proof}

\smallskip

\par\noindent{\em Remark}:
Corollary \ref{cor:laplacian-double-cone} is a generalization of the main result due to
Bose, Casaccino, Mancini and Severini \cite{bcms09} which studied Laplacian perfect state transfer 
in complete graphs with a missing edge. By viewing $K_{n} \setminus e$ as a double cone and 
using closure under complementation for Laplacian perfect state transfer, we found a simpler proof
for a more general result. 
Also, contrast Corollary \ref{cor:laplacian-double-cone} with a similar result in 
the adjacency matrix model due to Angeles-Canul \etal \cite{anoprt10}.
They showed that perfect state transfer occurs on $\comp{K_{2}} + H$ under 
a more complicated number-theoretic conditions and only when $H$ is regular.

\subsection{Double cones}

In this section, we show a tighter version of Corollary \ref{cor:laplacian-double-cone} 
using the machinery of quotient graphs relative to {\em almost equitable partitions}.
This provides a characterization of perfect state transfer on the double cones relative to the
standard Laplacian.
First, we state a symmetric quotient graph within the framework developed by Cardoso \etal \cite{cdr07}.

\begin{fact} \label{fact:equitable-laplacian} (based on Cardoso \etal \cite{cdr07}) \\
Let $G=(V,E)$ be a graph with an {almost} equitable partition $\pi$ given by $V = \bigcup_{k=1}^{m} V_{k}$
with constants $d_{j,k}$, for $1 \le j < k \le m$. Let $\NP$ be the (normalized) partition matrix of $\pi$, where
$\NP^{T}\NP = \Id_{m}$.
Then:
\begin{enumerate}
\item $\NP\NP^{T} = \diag(\{J_{|V_{k}|} : k \in [m]\})$ which commutes with $\DLp(G)$.
\item $\DLp(G)\NP = \NP B$, where $B$ is a $m \times m$ matrix defined as
\begin{equation}
B_{j,k} = 
	\left\{\begin{array}{ll}
	-\sqrt{d_{j,k}d_{k,j}} & \mbox{ if $j \neq k$ } \\
	\sum_{\ell \neq j} d_{j,\ell} & \mbox{ if $j = k$}
	\end{array}\right.
\end{equation}
where $j,k = 1,\ldots,m$.
\end{enumerate}
Thus, $B = \DLp(G/\pi)$ is the standard Laplacian of the quotient graph $G/\pi$. 
\end{fact}
\begin{proof}
Follows immediately from Cardoso \etal \cite{cdr07} via the normalized partition matrix.
\end{proof}

\medskip

\par\noindent
Using Fact \ref{fact:equitable-laplacian}, we provide the following tight version of 
Corollary \ref{cor:laplacian-double-cone}.

\begin{corollary} \label{cor:laplacian-double-cone-tight}
The join $\comp{K_{2}} + H$ has perfect state transfer at time $\pi/2$ between the vertices of
$K_{2}$ relative to the standard Laplacian 
{\em if and only if} $|V(H)| \equiv 2\pmod{4}$.
\end{corollary}
\begin{proof}
Let $G = \comp{K_{2}} + H$, where $H$ is a $n$-vertex graph.
By Fact \ref{fact:equitable-laplacian}, the Laplacian quotient $B$ of $G$ is given by
\begin{equation}
B = 
\begin{bmatrix}
n & -\sqrt{n} & 0 \\
-\sqrt{n} & 2 & -\sqrt{n} \\
0 & -\sqrt{n} & n
\end{bmatrix}
\end{equation}
In a quantum walk, up to phase factors, time reversal, and time dilations,
$B$ is equivalent to 
\begin{equation}
\tilde{B} =
\frac{-1}{\sqrt{n}}(B - n\Id)
=
\begin{bmatrix}
0 & 1 & 0 \\
1 & \frac{1}{\sqrt{n}}(n-2) & 1 \\
0 & 1 & 0
\end{bmatrix}
\end{equation}
We apply Fact \ref{fact:p3-optimal} with $\alpha = \frac{1}{\sqrt{n}}(n-2)$.
Note also that $\talpha = \alpha/2$ and $\Delta = \sqrt{\talpha^{2} + 2}$.
Thus, there is antipodal perfect state transfer using $\tilde{B}$ at time $t$ if and only if
\begin{equation}
\cos(\talpha t)\cos(\Delta t) = -1.
\end{equation}
This implies that $\talpha t, \Delta t \in \ZZ\pi$ and
$\talpha/\Delta$ is a rational number $p/q$ of distinct parities (either $p$ is odd and $q$ is even,
or $p$ is even and $q$ is odd). But, note that
\begin{equation}
\frac{\talpha}{\Delta}
=
\frac{n-2}{\sqrt{(n-2)^{2} + 8n}} = \frac{n-2}{n+2}.
\end{equation}
Since the parities of the numbers in the fraction $p/q$ must be distinct, it is clear that $n$ must be even
and must satisfy $n \equiv 2\pmod{4}$.
By Lemma \ref{lemma:lifting} (lifting), we obtain the claim on the double cone $\comp{K_{2}} + H$.
\end{proof}

\subsection{Joins}

We revisit perfect state transfer on graph joins relative to the standard Laplacian
and show a negative result on connected double cones.

\begin{fact} \label{fact:qwalk-join}
Let $G$ and $H$ be graphs on $m$ and $n$ vertices, respectively. 
For vertices $u$ and $v$ of $G$, the quantum walk on $G+H$ relative to the standard Laplacian
satisfies
\begin{equation}
\tbracket{\buket{u}}{e^{-it\DLp(G+H)}}{\buket{v}}
= 
e^{-itn}\tbracket{\buket{u}}{e^{-it\DLp(G)}}{\buket{v}}
+
\frac{(e^{-it(m+n)} - e^{-itn})}{m} 
+ 
\frac{(1 - e^{-it(m+n)})}{m+n}.
\end{equation}
Moreover, 
if perfect state transfer occurs between vertices $u$ and $v$ in $G+H$ at time $t$
relative to the standard Laplacian, then $t(m+n) \in 2\pi\ZZ$.
\end{fact}
\begin{proof}
Let the spectral decompositions of the standard Laplacians of $G$ and $H$ be
\begin{equation}
\DLp(G) = \sum_{k=0}^{m-1} \lambda_{k}E_{k},
\ \hspace{.2in} \
\DLp(H) = \sum_{\ell=0}^{n-1} \mu_{\ell}F_{\ell}.
\end{equation}
Then, the join $G+H$ has the following spectral decomposition:
\begin{equation}
\DLp(G+H)
=
0 \cdot z_{0}z_{0}^{T}
+ (m+n)z_{1}z_{1}^{T} 
+ \sum_{k=1}^{m-1} (n+\lambda_{k})\begin{bmatrix} E_{k} & O \\ O & O \end{bmatrix}
+ \sum_{\ell=1}^{n-1} (m+\mu_{\ell})\begin{bmatrix} O & O \\ O & F_{\ell} \end{bmatrix}
\end{equation}
where 
\begin{equation}
z_{0} = 
\frac{1}{\sqrt{m+n}}\one_{m+n},
\ \hspace{.2in} \
z_{1} = 
\frac{1}{\sqrt{mn(m+n)}}\begin{bmatrix} n\one_{m} \\[1.5ex] -m\one_{n} \end{bmatrix}.
\end{equation}

\smallskip

Using this, the quantum walk on $G+H$ relative to the standard Laplacian is 
\begin{equation} \label{eqn:qwalk-join}
e^{-it\DLp(G+H)}
=
\frac{J_{m+n}}{m+n}
+ e^{-it(m+n)}z_{1}z_{1}^{T} 
+ \sum_{k=1}^{m-1} e^{-it(n+\lambda_{k})} \begin{bmatrix} E_{k} & O \\ O & O \end{bmatrix}
+ \sum_{\ell=1}^{n-1} e^{-it(m+\mu_{\ell})} \begin{bmatrix} O & O \\ O & F_{\ell} \end{bmatrix}.
\end{equation}
For the vertices $u$ and $v$ of $G$, we have
\begin{eqnarray}
\tbracket{\buket{u}}{e^{-it\DLp(G+H)}}{\buket{v}}
	& = &
	\frac{1}{m+n} + \left[\frac{1}{m} - \frac{1}{m+n}\right] e^{-it(m+n)} 
		+ e^{-int}\sum_{k=1}^{m-1} e^{-it\lambda_{k}}\tbracket{\buket{u}}{E_{k}}{\buket{v}} \\
	& = &
	\frac{(e^{-it(m+n)} - e^{-itn})}{m} + \frac{(1 - e^{-it(m+n)})}{m+n}
		+ e^{-itn}\tbracket{\buket{u}}{e^{-it\DLp(G)}}{\buket{v}},
\end{eqnarray}
where we have used the fact $E_{0} = J_{m}/m$.

To show the second claim, let $y$ be a vertex of $H$. Then, by (\ref{eqn:qwalk-join}), we have
\begin{equation}
\tbracket{\buket{u}}{e^{-it\DLp(G+H)}}{\buket{y}}
=
\frac{1}{m+n}(1 - e^{-it(m+n)}).
\end{equation}
This expression is zero if perfect state transfer occurs between $u$ and $v$ in $G+H$.
Therefore, $e^{-it(m+n)} = 1$, which implies $t(m+n) \in 2\pi\ZZ$.
\end{proof}

We show that {\em connected} double cones $K_{2}+G$ have no perfect state transfer 
relative to the standard Laplacian, unlike its counterpart $\comp{K_{2}}+G$.

\begin{corollary}
For any graph $G$, there is no perfect state transfer on $K_{2} + G$ 
between the two vertices of $K_{2}$ relative to the standard Laplacian.
\end{corollary}
\begin{proof}
Let $G$ be a graph on $n$ vertices.
Suppose there is perfect state transfer on $K_{2}+G$ at time $t$ between 
the vertices $u$ and $v$ of $K_{2}$ relative to the standard Laplacian.
By Fact \ref{fact:qwalk-join}, we have 
\begin{equation} \label{eqn:connected-cone}
\tbracket{\buket{u}}{e^{-it\DLp(K_{2}+G)}}{\buket{v}}
= 
e^{-itn}
\left[\tbracket{\buket{u}}{e^{-it\DLp(K_{2})}}{\buket{v}}
+
\frac{1}{2}(e^{-2it} - 1)\right], 
\end{equation}
since $e^{-it(2+n)} = 1$. 
But, note that $\tbracket{\buket{u}}{e^{-it\DLp(K_{2})}}{\buket{v}} = \frac{1}{2}(1-e^{-2it})$.
So, the right-hand side of (\ref{eqn:connected-cone}) is zero, which is a contradiction 
since we assume $K_{2}+G$ has perfect state transfer at time $t$ between $u$ and $v$.
\end{proof}


\section{Signless Laplacian}

\subsection{Double cones}

By Fact \ref{fact:bipartite}, the quantum walks relative to the Laplacians $D \pm A$ are equivalent 
for regular and/or bipartite graphs. We describe a family of graphs with perfect state transfer 
relative to the signless Laplacian $D+A$, but not under the standard Laplacian $D-A$. 
First, we state some facts about quotient graphs relative to the signless Laplacian.

\begin{fact} \label{fact:equitable-signless}
Let $G=(V,E)$ be a graph with an equitable partition $\pi$ given by $V = \bigcup_{k=1}^{m} V_{k}$
with constants $d_{j,k}$, for $j,k=1,\ldots,m$. Let $\NP$ be the (normalized) partition matrix of $\pi$,
where $\NP^{T}\NP = \Id_{m}$. Then:
\begin{enumerate}
\item $\NP\NP^{T} = \diag(\{J_{|V_{k}|} : k \in [m]\})$ which commutes with $\SLp(G)$.
\item $\SLp(G)\NP = \NP B$, where $B$ is a $m \times m$ matrix defined as
\begin{equation} \label{eqn:signless-quotient}
B_{j,k} = 
	\left\{\begin{array}{ll}
	\sqrt{d_{j,k}d_{k,j}} & \mbox{ if $j \neq k$ } \\
	2d_{j,j} + \sum_{\ell \neq j} d_{j,\ell} & \mbox{ if $j = k$}
	\end{array}\right.
\end{equation}
where $j,k = 1,\ldots,m$.
\end{enumerate}
Thus, $B = \SLp(G/\pi)$ is the signless Laplacian of the quotient graph $G/\pi$.
\end{fact}
\begin{proof}
To show that $\NP\NP^{T}$ commutes with $\SLp(G)$, it suffices to show it commutes with $D(G)$.
Note that $D(G)$ is a diagonal matrix over $m$ blocks with the following form:
\begin{equation}
D(G) = \diag\left(\left\{\sum_{\ell=1}^{m} d_{k,\ell}\right\} \Id_{|V_{k}|} : k \in [m]\right).
\end{equation}
Since $\NP\NP^{T} = \diag(\{J_{|V_{k}|} : k \in [m]\})$, 
it commutes with $D(G)$ and hence with $\SLp(G)$.

Next, we show that $\SLp(G)\NP = \NP B$, where $B$ is given by (\ref{eqn:signless-quotient}).
If we let $B = \NP^{T}\SLp(G)\NP$,
\begin{eqnarray}
\tbracket{\buket{j}}{B}{\buket{k}}
	& = & \tbracket{\buket{j}}{\NP^{T}\SLp(G)\NP}{\buket{k}} \\
	& = & \frac{1}{\sqrt{|V_{j}||V_{k}|}}\sum_{u \in V_{j}, v \in V_{k}} \tbracket{\buket{u}}{\SLp(G)}{\buket{v}} \\
	& = & \frac{1}{\sqrt{|V_{j}||V_{k}|}}\sum_{u \in V_{j}, v \in V_{k}} \tbracket{\buket{u}}{(D(G) + A(G))}{\buket{v}} \\
	& = & \iverson{j=k}\left[2d_{j,j} + \sum_{\ell \neq j} d_{j,\ell}\right] 
			+ \iverson{j \neq k} \sqrt{d_{j,k}d_{k,j}}.
\end{eqnarray}
The above case for $j \neq k$ follows from Fact \ref{fact:equitable-adjacency}.

From $B = \NP^{T}\SLp(G)\NP$, by multiplying both sides by $\NP$ from the left, we get
\begin{equation}
\NP B = \NP\NP^{T}\SLp(G)\NP = \SLp(G)\NP\NP^{T}\NP = \SLp(G)\NP,
\end{equation}
since $\NP\NP^{T}$ commutes with $\SLp(G)$ and $\NP^{T}\NP = \Id_{m}$. This proves the second claim.
\end{proof}

\medskip

\begin{theorem} \label{thm:signless-double-cone}
For an integer $m \ge 1$, if $H$ is a $(2m,m-1)$-regular graph, 
then $\overline{K}_{2} + H$ has perfect state transfer relative to the signless Laplacian.
\end{theorem}
\begin{proof}
Let $H$ be a $(n,k)$-regular graph and consider the double cone $G = \overline{K}_{2} + H$.
By Fact \ref{fact:equitable-signless}, the signless Laplacian quotient $B$ of $G$ is given by
\begin{equation}
B =
\begin{bmatrix}
n & \sqrt{n} & 0 \\
\sqrt{n} & 2k+2 & \sqrt{n} \\
0 & \sqrt{n} & n
\end{bmatrix}
\end{equation}
If $k = (n-2)/2$, we have $B = n\Id + \sqrt{n}A(P_{3})$. 
Let $a$ and $b$ be the conical vertices of $G$ (which are also the antipodal vertices of the quotient).
Therefore,
\begin{equation}
\tbracket{\buket{b}}{e^{-itB}}{\buket{a}} = e^{-int} \tbracket{\buket{b}}{e^{-i(\sqrt{n}t)A(P_{3})}}{\buket{a}},
\end{equation}
which implies that the signless Laplacian quotient of $G$ has perfect state transfer
at time $t = \pi/\sqrt{2n}$.

Using Lemma \ref{lemma:lifting}, we {\em lift} the perfect state transfer from the quotient 
to the original graph: 
\begin{equation}
\tbracket{\buket{b}}{e^{-it\SLp(G)}}{\buket{a}} 
= 
\tbracket{\buket{b}}{\NP^{T}e^{-it\SLp(G)}\NP}{\buket{a}}
= 
\tbracket{\buket{b}}{e^{-itB}}{\buket{a}},
\end{equation}
since $e^{-itB} = \NP^{T}e^{-it\SLp(G)}\NP$.
Finally, we set $n = 2m$ to complete the proof.
\end{proof}

\medskip
\par\noindent{\em Example}:
For completeness, we describe a simple family $\{G_{m}\}$ of $(2m,m-1)$-regular graphs, 
where the double cone $\comp{K_{2}} + G_{m}$ has perfect state transfer between the
two conical vertices relative to the signless Laplacian (by Theorem \ref{thm:signless-double-cone}).
Each graph $G_{m}$ is a circulant over $\ZZ_{2m}$ with the following generating set:
\begin{equation}
S_{m} 
= 
	\left\{\begin{array}{ll}
	\{\pm1,\ldots,\pm \textstyle \frac{1}{2}(m-1)\} & \mbox{ if $m-1$ is even} \\
	\{\pm1,\ldots,\pm \textstyle \frac{1}{2}(m-2)\} \cup \{\pm m\} & \mbox{ if $m-1$ is odd}
	\end{array}\right.
\end{equation}
It is clear that the circulant $G_{m} = \Circ(\ZZ_{2m},S_{m})$ is $(2m,m-1)$-regular. 
Moreover, note that the double cone $\comp{K_{2}}+G_{m}$ has no perfect state transfer
relative to the standard Laplacian, whenever $2m \equiv 0\pmod{4}$ 
(by Corollary \ref{cor:laplacian-double-cone}).

\subsection{Line graphs}

The (normalized) {\em incidence matrix} $\NB$ of $G$ is a $n \times m$ matrix defined as 
$\NB_{u,e} = \frac{1}{\sqrt{2}}\iverson{u \in e}$, for each vertex $u \in V(G)$ 
and each edge $e \in E(G)$.
We state a well-known connection between the signless Laplacian and line graphs.

\begin{fact} \label{fact:signless-line} 
For any graph $G=(V,E)$ with (normalized) incidence matrix $\NB$, we have:
\begin{enumerate}[(i)]
\item $\NB\NB^{T} = \frac{1}{2}\SLp$,
	and is nonsingular if $G$ is connected and nonbipartite.
\item $\NB^{T}\NB = \frac{1}{2}A(\Line(G)) + \Id$,
	and is nonsingular if $G$ is a tree.
\end{enumerate}
\end{fact}
\begin{proof}
See Godsil and Royle \cite{godsil-royle01}, Theorem 8.2.1 for example.
\end{proof}

\ignore{
\begin{fact} \label{fact:incidence-nonsingular}
Let $G$ be a connected graph on $n$ vertices with incidence matrix $\NB$. Then:
\begin{enumerate}[(i)]
\item If $G$ is nonbipartite, then $\NB\NB^{T}$ is nonsingular.
\item If $G$ is a tree, then $\NB^{T}\NB$ is nonsingular.
\end{enumerate}
\end{fact}
\begin{proof}
We know that $\rk(B) = n - \iverson{G \mbox{~{\em bipartite}}}$ 
(see Theorem 8.2.1 in Godsil and Royle \cite{godsil-royle01}).
If $G$ is nonbipartite, then $\rk(B) = n$.
Since $\NB$ has $n$ nonzero singular values, all $n$ eigenvalues of $\NB\NB^{T}$ are nonzero.
If $G$ is a tree, then $\rk(B) = n-1$. 
Thus, $\NB$ has $n-1$ nonzero singular values and all $n-1$ eigenvalues of $\NB^{T}\NB$ are nonzero.
\end{proof}
}

\medskip
We observe the following connections 
between the quantum walk on a graph $G$ relative to the signless Laplacian 
and the quantum walk on the line graph $\Line(G)$ relative to the adjacency matrix.
The third observation in Lemma \ref{lemma:SignlessLine} was suggested by Ada Chan.

\begin{lemma} \label{lemma:SignlessLine}
Let $G$ be a graph with (normalized) incidence matrix $\NB$. 
Then:
\begin{enumerate}[a)]
\item $\NB^{T}e^{-it\SLp(G)} = e^{-2it} e^{-itA(\Line(G))}\NB^{T}$.
\item $e^{-it\SLp(G)}\NB = e^{-2it} \NB e^{-itA(\Line(G))}$.
\item $\NB^{T}e^{-it\SLp(G)}\NB = e^{-2it} e^{-itA(\Line(G))}\NB^{T}\NB$.
\end{enumerate}
\end{lemma}
\begin{proof}
For the first identity, we have 
\begin{equation}
\NB^{T}e^{-it\SLp(G)}
	= \NB^{T} \sum_{k=0}^{\infty} \frac{(-2it)^{k}}{k!}(\NB\NB^{T})^{k}
	= \sum_{k=0}^{\infty} \frac{(-2it)^{k}}{k!}(\NB^{T}\NB)^{k}\NB^{T}
	= e^{-2it\NB^{T}\NB}\NB^{T},
\end{equation}
and apply $2\NB^{T}\NB = A(\Line(G)) + 2\Id$ to get the result.
The proof of the second identity is similar.

For the third identity, we have
\begin{equation}
\NB^{T} e^{-it\SLp(G)} \NB
	= \NB^{T} \left[ \sum_{k=0}^{\infty} \frac{(-2it)^{k}}{k!} (\NB\NB^{T})^{k} \right] \NB 
	= \left[ \sum_{k=0}^{\infty} \frac{(-2it)^{k}}{k!} (\NB^{T}\NB)^{k} \right] \NB^{T}\NB, 
\end{equation}
and again apply $2\NB^{T}\NB = A(\Line(G)) + 2\Id$.
\end{proof}

\begin{theorem} \label{thm:pst-signless-line}
Let $G$ be a graph which has perfect state transfer at time $t$ 
from a vertex $u_{1}$ of degree one to another vertex $u_{2}$ 
relative to the signless Laplacian.
Then, $u_{2}$ must have degree one and the line graph $\Line(G)$ 
has perfect state transfer at time $t$ between the unique edges incident to $u_{1}$ and $u_{2}$ 
relative to the adjacency matrix.
\end{theorem}
\begin{proof}
Let $\NB$ be the normalized incidence matrix of $G$. 
Suppose $G$ has perfect state transfer from vertex $u_{1}$ of degree one to another
vertex $u_{2}$ at time $t$ relative to the signless Laplacian.

Let $e_{1}$ be the unique edge incident to $u_{1}$
and let $e_{2}$ be any edge incident to $u_{2}$. 
Say, $e_{2} = (u_{2},z)$, for some vertex $z$.
By Lemma \ref{lemma:SignlessLine}, we have 
\begin{equation}
e^{-2it}e^{-itA(\Line(G))}\NB^{T}
=
\NB^{T} e^{-it\SLp(G)}. 
\end{equation}
Therefore, up to phase factors, we have 
\begin{eqnarray}
\textstyle \frac{1}{\sqrt{2}} \tbracket{\buket{e_{2}}}{e^{-itA(\Line(G))}}{\buket{e_{1}}}
	& = & \tbracket{\buket{e_{2}}}{e^{-itA(\Line(G))}\NB^{T}}{\buket{u_{1}}} \\
	& = & \tbracket{\buket{e_{2}}}{\NB^{T} e^{-it\SLp(G)}}{\buket{u_{1}}} \\
	& = & \textstyle \frac{1}{\sqrt{2}} \tbracket{\buket{u_{2}} + \buket{z}}{e^{-it\SLp(G)}}{\buket{u_{1}}}.
\end{eqnarray}
We have $\tbracket{\buket{z}}{e^{-it\SLp(G)}}{\buket{u_{1}}} = 0$
since $|\tbracket{\buket{u_{2}}}{e^{-it\SLp(G)}}{\buket{u_{1}}}| = 1$.
Thus, 
\begin{equation}
\tbracket{\buket{e_{2}}}{e^{-itA(\Line(G))}}{\buket{e_{1}}}
=
\tbracket{\buket{u_{2}}}{e^{-it\SLp(G)}}{\buket{u_{1}}},
\end{equation}
which completes the proof.
\end{proof}

\medskip

We show that Theorem \ref{thm:pst-signless-line} may be used as a tool for showing the
absence of perfect state transfer relative to the signless Laplacian in some graphs. 
It would be interesting if we can apply this similarly in the other direction.

\medskip

\begin{corollary} \label{cor:path-no-signless-pst}
For $n \ge 5$, there is no antipodal perfect state transfer on $P_{n}$ relative to the signless Laplacian.
\end{corollary}
\begin{proof}
For $n \ge 5$, suppose there is a time $t$ so $|\tbracket{\buket{1}}{e^{-it\SLp(P_{n})}}{\buket{n}}| = 1$.
By Theorem \ref{thm:pst-signless-line}, since $\Line(P_{n}) = P_{n-1}$, we have
$|\tbracket{\buket{e_{1}}}{e^{-itA(P_{n-1})}}{\buket{e_{n-1}}}| = 1$,
where $e_{1}$ and $e_{n-1}$ are the edges incident to vertices $1$ and $n$, respectively.
But, there is no perfect state transfer on paths $P_{m}$, for $m \ge 4$ 
(see Christandl \etal \cite{cdel04}). 
\end{proof}

\medskip
\par\noindent{\em Remark}:
A better version of Corollary \ref{cor:path-no-signless-pst} (with optimal proof) is due to Godsil 
who showed that there is no perfect state transfer on $P_{n}$, for $n \ge 3$, relative to the standard Laplacian.
So, a minor novelty of Corollary \ref{cor:path-no-signless-pst} is in using the spectral link 
between the unnormalized Laplacians and the adjacency matrix of the line graph.
The latest breakthrough result by Coutinho and Liu \cite{cl14} showed that Godsil's result holds 
for trees with at least three vertices.

\subsection{Odd unicyclic graphs}

We describe an application of Theorem \ref{thm:pst-signless-line} to {nonbipartite} graphs. 
To this end, we consider a family of graphs obtained from paths by adding a unique odd-cycle 
(here, we focus on the three-cycle $C_{3}$). 
We show that this family of odd unicyclic graphs has no antipodal perfect state transfer 
relative to the signless Laplacian.
Our proof exploits a connection between perfect state transfer and {\em controllability}
described by Godsil \cite{godsil-ac12,godsil-ejla12} 
(see also Godsil and Severini \cite{gs10}).

We formally define our family of odd unicyclic graphs.
For an integer $m \ge 1$, let $\OU_{m}$ be the graph obtained by attaching two pendant paths 
$P_{m+1}$ (with $m$ edges) to a three-cycle $C_{3}$ 
(see Figure \ref{fig:controllable}(a)). 
The line graph of $\OU_{m}$ is a graph which has two pendant paths $P_{m}$
attached to the pair of vertices of degree two in the cone $K_{1} + P_{4}$.

In what follows, we briefly describe the machinery of controllable subsets on graphs.
Let $G = (V,E)$ be a graph on $n$ vertices with adjacency matrix $A$.
For a subset of vertices $S \subseteq V$, the {\em walk matrix} $W_{S}$ on $G$ with respect to $S$
is defined as
\begin{equation}
W_{S} =
\begin{bmatrix}
e_{S} & Ae_{S} & A^{2}e_{S} & \ldots & A^{n-1}e_{S} 
\end{bmatrix}
\end{equation}
where $e_{S}$ denotes the characteristic vector of $S$.
We say that the pair $(G,S)$ is {\em controllable} if $W_{S}$ has full rank.
A vertex $u$ of $G$ is called controllable if $(G,\{u\})$ is controllable.

\begin{figure}[t]
\begin{center}
\ignore{
\begin{tikzpicture}
%
\node at (0,0)[scale=0.5]{};

\draw (-1.5,0)--(1.5,0);
\draw (-0.5,0)--(0,1);
\draw (+0.5,0)--(0,1);
\node at (-1.5,0)[circle, fill=black][scale=0.5]{};
\node at (+1.5,0)[circle, fill=black][scale=0.5]{};
\node at (0,+1)[circle, fill=black][scale=0.5]{};
\node at (-0.5,0)[circle, fill=black][scale=0.5]{};
\node at (+0.5,0)[circle, fill=black][scale=0.5]{};

\node at (0,1.3)[scale=0.9]{};
\node at (-1.5,-0.3)[scale=0.9]{};
\node at (-0.5,-0.3)[scale=0.9]{};
\node at (+0.5,-0.3)[scale=0.9]{};
\node at (+1.5,-0.3)[scale=0.9]{};
\end{tikzpicture}
\quad
\begin{tikzpicture}
%
\node at (0,0)[scale=0.5]{};

\draw (-1.5,0)--(1.5,0);
\draw (-0.5,0)--(0,1);
\draw (-1.5,0)--(0,1);
\draw (+0.5,0)--(0,1);
\draw (+1.5,0)--(0,1);
\node at (-1.5,0)[circle, fill=black][scale=0.5]{};
\node at (+1.5,0)[circle, fill=black][scale=0.5]{};
\node at (0,+1)[circle, fill=black][scale=0.5]{};
\node at (-0.5,0)[circle, fill=black][scale=0.5]{};
\node at (+0.5,0)[circle, fill=black][scale=0.5]{};

\node at (0,1.3)[scale=0.9]{};
\node at (-1.5,-0.3)[scale=0.9]{};
\node at (-0.5,-0.3)[scale=0.9]{};
\node at (+0.5,-0.3)[scale=0.9]{};
\node at (+1.5,-0.3)[scale=0.9]{};
\end{tikzpicture}
\quad
}
\begin{tikzpicture}
%
\node at (0,0)[scale=0.5]{};

\draw (-2.5,0)--(2.5,0);
\draw (-0.5,0)--(0,1);
\draw (+0.5,0)--(0,1);
\node at (-2.5,0)[circle, fill=black][scale=0.5]{};
\node at (-1.5,0)[circle, fill=black][scale=0.5]{};
\node at (+1.5,0)[circle, fill=black][scale=0.5]{};
\node at (+2.5,0)[circle, fill=black][scale=0.5]{};
\node at (0,+1)[circle, fill=black][scale=0.5]{};
\node at (-0.5,0)[circle, fill=black][scale=0.5]{};
\node at (+0.5,0)[circle, fill=black][scale=0.5]{};

\node at (0,1.3)[scale=0.9]{};
\node at (-2.5,-0.3)[scale=0.9]{};
\node at (-1.5,-0.3)[scale=0.9]{};
\node at (-0.5,-0.3)[scale=0.9]{};
\node at (+0.5,-0.3)[scale=0.9]{};
\node at (+1.5,-0.3)[scale=0.9]{};
\node at (+2.5,-0.3)[scale=0.9]{};
\end{tikzpicture}
\quad \quad \quad \quad
\begin{tikzpicture}
%
\node at (0,0)[scale=0.5]{};

\draw (-2.5,0)--(2.5,0);
\draw (-0.5,0)--(0,1);
\draw (-1.5,0)--(0,1);
\draw (+0.5,0)--(0,1);
\draw (+1.5,0)--(0,1);
\node at (-2.5,0)[circle, fill=black][scale=0.5]{};
\node at (-1.5,0)[circle, fill=black][scale=0.5]{};
\node at (+1.5,0)[circle, fill=black][scale=0.5]{};
\node at (+2.5,0)[circle, fill=black][scale=0.5]{};
\node at (0,+1)[circle, fill=black][scale=0.5]{};
\node at (-0.5,0)[circle, fill=black][scale=0.5]{};
\node at (+0.5,0)[circle, fill=black][scale=0.5]{};

\node at (0,1.3)[scale=0.9]{};
\node at (-2.5,-0.3)[scale=0.9]{};
\node at (-1.5,-0.3)[scale=0.9]{};
\node at (-0.5,-0.3)[scale=0.9]{};
\node at (+0.5,-0.3)[scale=0.9]{};
\node at (+1.5,-0.3)[scale=0.9]{};
\node at (+2.5,-0.3)[scale=0.9]{};
\end{tikzpicture}
\caption{The odd unicyclic graph $\OU_{m}$ is obtained by attaching two pendant paths $P_{m+1}$ 
to a three-cycle $C_{3}$. The example shows $\OU_{2}$ (on left) and its line graph $\Line(\OU_{2})$ (on right).
}
\label{fig:ou2}
\vspace{.2in}
\hrule
\end{center}
\end{figure}
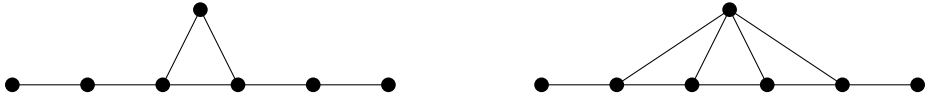

\medskip

The following theorems of Godsil on controllability and state transfer will prove useful.

\begin{theorem} \label{thm:godsil-pst-controllable} (Godsil \cite{godsil-ejla12}, Theorem 7.4) \\
If $G$ has perfect state transfer (relative to the adjacency matrix) between vertices $u$ and $v$, 
then neither $u$ nor $v$ is controllable.
\end{theorem}

\begin{theorem} \label{thm:godsil-cone-controllable} (Godsil \cite{godsil-ac12}) \\
Let $G = (V,E)$ be a graph and $S \subseteq V$ be a subset of vertices.
Let $\widehat{G}_{S}$ be a graph obtained from $G$ and a path with endpoints $u$ and $v$
(which may be identical) whereby we connect $u$ to all vertices in $S$.
If $(G,S)$ is controllable, then $v$ is controllable in $\widehat{G}_{S}$.
\end{theorem}

A main ingredient of our proof is the next lemma on the controllability of $K_{1} + P_{4}$.

\begin{lemma} \label{lemma:fan-pendant-path}
Let $\Fan = K_{1} + P_{4}$ and let $u$ and $v$ be vertices of degree two in $\Fan$.
For $m \ge 0$, let $\GG_{m}$ be the graph obtained by attaching a pendant path $P_{m+1}$ 
(see Figure \ref{fig:controllable}(b)) to vertex $v$.
Then, $u$ is controllable in $\GG_{m}$ if and only if $m \not\equiv 2\pmod{3}$.
\end{lemma}
\begin{proof}
Let $A$ be the adjacency matrix of $\GG_{m}$ whose spectral decomposition is given by
\begin{equation}
A = \sum_{\ell=1}^{n} \lambda_{\ell}\ket{z_{\ell}}\bra{z_{\ell}}
\end{equation}
where $\{\ket{z_{1}},\ldots,\ket{z_{n}}\}$ is the set of orthonormal eigenvectors of $A$
which satisfies $A\ket{z_{\ell}} = \lambda_{\ell}\ket{z_{\ell}}$, for every $\ell=1,\ldots,n$.
Consider the walk matrix $W_{u}$ relative to vertex $u$:
\begin{equation}
W_{u} = \sum_{k=0}^{n-1} A^{k}\uket{u}\tbra{k}
\end{equation}
The rank of $W_{u}$ is equal to the cardinality of the set 
$\{\ell : \bracket{\bvket{z_{\ell}}}{\buket{u}} \neq 0\}$.
To see this, note that
\begin{equation}
A^{k}\uket{u} = \sum_{\ell=1}^{n} \lambda_{\ell}^{k}\bracket{\bvket{z_{\ell}}}{\buket{u}}\ket{z_{\ell}},
\end{equation}
which shows that the columns of $W_{u}$ is spanned by the vectors $\ket{z_{\ell}}$ satisfying
$\bracket{\bvket{z_{\ell}}}{\buket{u}} \neq 0$.

In what follows, we label the vertices of $\Fan$ as $\{0,1,2,3,4\}$ where $0$ is the conical
vertex, $1$ and $4$ are the endpoints of $P_{4}$, $2$ and $3$ are the middle vertices with
$2$ adjacent to $1$ and $3$ adjacent to $4$. The vertices of the pendant path $P_{m+1}$ will
be labeled consecutively as $4$ followed by $4+k$, for $k = 1,\ldots,m$;
see Figure \ref{fig:controllable}(b).
We show $\bracket{\bvket{z_{\ell}}}{\buket{1}} = 0$, for some $\ell$, if and only if $m \equiv 2\pmod{3}$.

Suppose $\ket{z}$ is an eigenvector of $A$ with eigenvalue $\lambda$ 
where $\bracket{\ket{z}}{\buket{1}}=0$.
Assume, without loss of generality, that 
$\bracket{\bvket{z}}{\buket{2}} = a$, for some $a \neq 0$.
Using the fact that 
$\lambda\bracket{\bvket{z}}{\buket{u}} = \sum_{v \sim u} \bracket{\bvket{z}}{\buket{v}}$
and following the chain of implications, we obtain:
\begin{eqnarray}
\bracket{\bvket{z}}{\buket{0}} & = & -a \\
\bracket{\bvket{z}}{\buket{3}} & = & (1+\lambda)a \\
\bracket{\bvket{z}}{\buket{4}} & = & \lambda(1+\lambda)a.
\end{eqnarray}
Since 
$\lambda\bracket{\bvket{z}}{\buket{0}} = \sum_{k=1}^{4} \bracket{\bvket{z}}{\buket{k}}$,
we have $(1 + \lambda)(2 + \lambda) = 0$, which implies $\lambda = -1$ or $\lambda = -2$. 
We consider these two cases separately.

\begin{figure}[t]
\begin{center}
\begin{tikzpicture}
%
\node at (0,0)[scale=0.5]{};

\draw (-1.5,0)--(1.5,0);
\draw (-0.5,0)--(0,1);
\draw (-1.5,0)--(0,1);
\draw (+0.5,0)--(0,1);
\draw (+1.5,0)--(0,1);
\node at (-1.5,0)[circle, fill=white][scale=0.5]{};
		\draw[line width=0.2mm] (-1.5,0) circle (0.09cm);
\node at (+1.5,0)[circle, fill=black][scale=0.5]{};
\node at (0,+1)[circle, fill=black][scale=0.5]{};
\node at (-0.5,0)[circle, fill=black][scale=0.5]{};
\node at (+0.5,0)[circle, fill=black][scale=0.5]{};

\node at (0,1.3)[scale=0.9]{$0$};
\node at (-1.5,-0.3)[scale=0.9]{$1$};
\node at (-0.5,-0.3)[scale=0.9]{$2$};
\node at (+0.5,-0.3)[scale=0.9]{$3$};
\node at (+1.5,-0.3)[scale=0.9]{$4$};
\end{tikzpicture}
\quad \quad \quad \quad
\begin{tikzpicture}
%
\node at (0,0)[scale=0.5]{};

\draw (-1.5,0)--(4.5,0);
\draw (-1.5,0)--(0,1);
\draw (-0.5,0)--(0,1);
\draw (+0.5,0)--(0,1);
\draw (+1.5,0)--(0,1);
\node at (-1.5,0)[circle, fill=white][scale=0.5]{};
		\draw[line width=0.2mm] (-1.5,0) circle (0.09cm);
\node at (0,+1)[circle, fill=black][scale=0.5]{};
\node at (-0.5,0)[circle, fill=black][scale=0.5]{};
\node at (+0.5,0)[circle, fill=black][scale=0.5]{};
\node at (+1.5,0)[circle, fill=black][scale=0.5]{};
\node at (+2.5,0)[circle, fill=black][scale=0.5]{};
\node at (+3.5,0)[circle, fill=black][scale=0.5]{};
\node at (+4.5,0)[circle, fill=black][scale=0.5]{};

\node at (0,1.3)[scale=0.9]{$0$};
\node at (-1.5,-0.3)[scale=0.9]{$1$};
\node at (-0.5,-0.3)[scale=0.9]{$2$};
\node at (+0.5,-0.3)[scale=0.9]{$3$};
\node at (+1.5,-0.3)[scale=0.9]{$4$};
\node at (+2.5,-0.3)[scale=0.9]{$5$};
\node at (+3.5,-0.3)[scale=0.9]{$6$};
\node at (+4.5,-0.3)[scale=0.9]{$7$};
\end{tikzpicture}

\caption{Small graphs with controllable vertices:
(a) The graph $\Fan$ is the cone $K_{1} + P_{4}$ (vertex $1$ is controllable; so is vertex $4$).
(b) The graph $\Fan$ with a pendant path $P_{m+1}$:
	vertex $1$ is controllable if and only if $m \not\equiv 2\pmod{3}$ (see Lemma \ref{lemma:fan-pendant-path}).
	The example shows $m = 3$.
}
\label{fig:controllable}
\vspace{.2in}
\hrule
\end{center}
\end{figure}
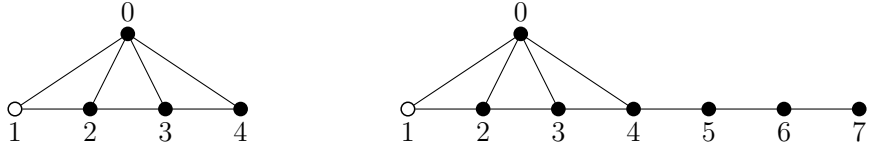

{\em Case}: $\lambda = -1$.
We have $\bracket{\bvket{z}}{\buket{3}} = \bracket{\bvket{z}}{\buket{4}} = 0$. 
For $k \ge 0$, this forces the three-step sequence along the pendant path:
\begin{equation}
\bracket{\bvket{z}}{\buket{4+k}} = 
\left\{\begin{array}{ll}
0 & \mbox{ if $k \equiv 0\pmod{3}$ } \\
+a & \mbox{ if $k \equiv 1\pmod{3}$ } \\
-a & \mbox{ if $k \equiv 2\pmod{3}$ }
\end{array}\right.
\end{equation}
The last vertex on the pendant path must have index $k \equiv 2\pmod{3}$ for the eigenvector $\ket{z}$
to be well-defined.

{\em Case}: $\lambda = -2$.
We have $\bracket{\bvket{z}}{\buket{3}} = -a$ and $\bracket{\bvket{z}}{\buket{4}} = 2a$.
For $k \ge 0$, this forces the following pattern along the pendant path:
\begin{equation}
\bracket{\bvket{z}}{\buket{4+k}} = (-1)^{k}2a,
\end{equation}
which must continue indefinitely. Hence, such an eigenvector $\ket{z}$ does not exist.

This proves the claim.
\end{proof}

We state our other corollary of Theorem \ref{thm:pst-signless-line} for the family 
of odd unicyclic of graphs $\OU_{m}$. This provides a nonbipartite generalization of
Corollary \ref{cor:path-no-signless-pst}.

\begin{corollary} \label{cor:no-signless-pst-odd-unicyclic}
For $m \ge 1$, the family of graphs $\OU_{m}$ has no antipodal perfect state transfer under 
the signless Laplacian whenever $m \not\equiv 0\pmod{3}$.
\end{corollary}
\begin{proof} 
Let $\Fan$ be the cone $K_{1} + P_{4}$ and let $u$ and $v$ be the vertices of degree two in $\Fan$.
Also, let $\GG_{m}$ denote the graph obtained by attaching to $v$ a pendant path with $m$ edges.
Then, the line graph of $\OU_{m}$, that is $\Line(\OU_{m})$, is obtained from $\GG_{m}$ by attaching
to $u$ a pendant path with $m$ edges.

To prove the claim, we show that $\Line(\OU_{m})$ does not have antipodal perfect state transfer under
the adjacency matrix and then apply Theorem \ref{thm:pst-signless-line}. 
By Theorem \ref{thm:godsil-pst-controllable}, it suffices to show that the two ``antipodal'' vertices 
of minimum degree in $\Line(\OU_{m})$ are controllable.
By Lemma \ref{lemma:fan-pendant-path}, the two vertices of degree two in the cone $\Fan$ are controllable.
Using Theorem \ref{thm:godsil-cone-controllable}, we conclude that the unique vertex of degree one in 
$\GG_{m}$ is controllable.
Since vertex $u$ in $\GG_{m}$ is controllable (by Lemma \ref{lemma:fan-pendant-path} again), 
if we attach a pendant path with $m$ edges to $u$, the other endpoint of this path is controllable
by Theorem \ref{thm:godsil-cone-controllable}.
\end{proof}

\par\noindent{\em Remark}:
Our argument in Corollary \ref{cor:no-signless-pst-odd-unicyclic} allows pendant paths of different 
lengths attached to a three-cycle provided the length of one of the paths is not divisible by three.
It would be interesting to show a similar result to Corollary \ref{cor:no-signless-pst-odd-unicyclic}
for arbitrary odd unicyclic graphs. These graphs are interesting since $\NB^{T}\NB$ is 
nonsingular\footnote{Doob \cite{d73} showed that $-2 \in \Sp(\Line(G))$ if and only if
$G$ contains an even cycle or two odd cycles in the same component.}.


\section{Normalized Laplacian}

A quantum walk on the hypercube $Q_{n}$ relative to the adjacency matrix has antipodal perfect state transfer 
at time $\pi/2$ for {\em any} $n$. This might contradict the postulate that the speed of light is constant.
In contrast, a quantum walk on $Q_{n}$ relative to the normalized Laplacian has antipodal perfect state transfer 
at time $n\pi/2$. 
Thus, the normalized Laplacian takes into account the diameter of the $n$-cube whereas the adjacency matrix does not.
This motivates a closer study of quantum walks relative to normalized Laplacians.

\subsection{Weak products}

We show that the weak product is a useful operation for constructing classes of
graphs with perfect state transfer relative to the normalized Laplacian.
First, we observe that the normalized Laplacian of a weak product $G \times H$ has a
strong resemblance in form to the adjacency matrix of a strong product $G \boxtimes H$.

\begin{fact} \label{fact:normalized-weak-product}
For graphs $G$ and $H$, we have
\begin{equation}
\NLp(G \times H) = \NLp(G) \otimes \Id_{H} + \Id_{G} \otimes \NLp(H) - \NLp(G) \otimes \NLp(H).
\end{equation}
\end{fact}
\begin{proof}
Note that the degree matrix of $G \times H$ is given by $D(G \times H) = D(G) \otimes D(H)$.
The normalized Laplacian is (also) defined as $\NLp = \Id - \NA$, where
$\NA = D^{-1/2}AD^{-1/2}$ is the normalized adjacency matrix. In our case, we have
$\NA(G \times H) = \NA(G) \otimes \NA(H)$.
Therefore,
\begin{eqnarray}
\NLp(G \times H)
	& = & \Id_{G} \otimes \Id_{H} - \NA(G) \otimes \NA(H) \\
	& = & \Id_{G} \otimes \Id_{H} - (\Id_{G} - \NLp(G)) \otimes (\Id_{H} - \NLp(H)) \\
	& = & \NLp(G) \otimes \Id_{H} + \Id_{G} \otimes \NLp(H) - \NLp(G) \otimes \NLp(H).
\end{eqnarray}
This proves the claim.
\end{proof}

\medskip
We apply Fact \ref{fact:normalized-weak-product} to derive a useful form on the quantum walk
on a weak product relative to the normalized Laplacian.

\begin{lemma} \label{lemma:normalized-qwalk}
Let $G$ and $H$ be graphs whose normalized Laplacians have spectral decompositions
given by $\NLp(G) = \sum_{k} \lambda_{k} E_{k}$ 
and $\NLp(H) = \sum_{\ell} \mu_{\ell} F_{\ell}$.
Then, the quantum walk on $G \times H$ relative to the normalized Laplacian is given by
\begin{equation} \label{eqn:qwalk-normalized}
\exp(-it\NLp(G \times H))
=
\sum_{k,\ell} \exp\left[-it\left(\lambda_{k} + \mu_{\ell} - \lambda_{k}\mu_{\ell}\right)\right]
E_{k} \otimes F_{\ell}.
\end{equation}
\end{lemma}
\begin{proof}
Follows from Fact \ref{fact:normalized-weak-product} since $\NLp(G \times H)$ consists
of three commuting matrices $\NLp(G) \otimes \Id_{H}$, $\Id_{G} \otimes \NLp(H)$, and
$\NLp(G) \otimes \NLp(H)$.
\end{proof}

\smallskip

We show a closure property for perfect state transfer under weak products relative to
the normalized Laplacian.

\begin{theorem} \label{thm:normalized-weak-pst1}
Let $G$ be a graph with perfect state transfer between vertices $g_{1}$ and $g_{2}$ 
at time $t_{G}$ relative to the normalized Laplacian. 
Suppose that $H$ is a graph where
\begin{equation}
t_{G} \Sp_{\NLp}(H)(\Sp_{\NLp}(G) - 1) \subseteq 2\pi\ZZ.
\end{equation} 
Then, $G \times H$ has perfect state transfer between vertices $(g_{1},h)$ and $(g_{2},h)$
at time $t_{G}$ relative to the normalized Laplacian.
\end{theorem}
\begin{proof}
Suppose $\NLp(G) = \sum_{k} \lambda_{k} E_{k}$ and $\NLp(H) = \sum_{\ell} \mu_{\ell} F_{\ell}$
are the spectral decompositions of the normalized Laplacians of $G$ and $H$.
By Lemma \ref{lemma:normalized-qwalk}, we have 
\begin{equation}
\tbracket{\buket{(g_{2},h_{2})}}{e^{-it\NLp(G \times H)}}{\buket{(g_{1},h_{1})}}
=
\sum_{k} e^{-it\lambda_{k}} \tbracket{\buket{g_{2}}}{E_{k}}{\buket{g_{1}}} 
	\sum_{\ell} e^{-it\mu_{\ell}(1 - \lambda_{k})} \tbracket{\buket{h_{2}}}{F_{\ell}}{\buket{h_{1}}}.
\end{equation}
Note we have used $\buket{(g,h)} = \buket{g} \otimes \buket{h}$.
Suppose at time $t_{G}$, we have $|\tbracket{\buket{g_{2}}}{e^{-it_{G}\NLp(G)}}{\buket{g_{1}}}| = 1$.
Since $t_{G}\Sp_{\NLp}(H)(\Sp_{\NLp}(G)-1) \subseteq 2\pi\ZZ$, we have
\begin{eqnarray}
\tbracket{\buket{(g_{2},h_{2})}}{e^{-it_{G}\NLp(G \times H)}}{\buket{(g_{1},h_{1})}}
	& = & \sum_{k} e^{-it_{G}\lambda_{k}} \tbracket{\buket{g_{2}}}{E_{k}}{\buket{g_{1}}} 
			\sum_{\ell} \tbracket{\buket{h_{2}}}{F_{\ell}}{\buket{h_{1}}} \\
	& = & \tbracket{\buket{g_{2}}}{e^{-it_{G}\NLp(G)}}{\buket{g_{1}}} \bracket{\buket{h_{2}}}{\buket{h_{1}}}, 
\end{eqnarray}
which proves the claim.
\end{proof}

\smallskip

\begin{corollary}
For any integer $m \ge 1$, the weak product $P_{3} \times K_{2m}$
has perfect state transfer at time $t = (2m-1)\pi$
relative to the normalized Laplacian.
\end{corollary}
\begin{proof}
The normalized Laplacian spectrum of the clique $K_{2m}$ is given by
\begin{equation}
\Sp_{\NLp}(K_{2m}) = \left\{0, 1 + \frac{1}{(2m-1)}\right\}. 
\end{equation}
Let $t_{G} = (2m-1)\pi$. 
By Fact \ref{fact:normalized-p3}, the spectrum of $P_{3}$ is given by
$\Sp_{\NLp}(P_{3}) = \{0,1,2\}$ and it has perfect state transfer at time $t_{G}$
relative to the normalized Laplacian.
Note that
\begin{equation}
(2m-1)\pi \times \left\{0,1 + \frac{1}{(2m-1)}\right\} \times \{0,1,2\} \subseteq 2\pi\ZZ.
\end{equation}
Thus, by Theorem \ref{thm:normalized-weak-pst1}, 
$P_{3} \times K_{2m}$ has perfect state transfer at time $t_{G}$ relative to the normalized Laplacian.
\end{proof}

\smallskip

\begin{corollary}
For any integer $m \ge 1$, the weak product $P_{3} \times Q_{2m-1}$
has perfect state transfer at time $t = (2m-1)\pi$
relative to the normalized Laplacian.
\end{corollary}
\begin{proof}
The normalized Laplacian spectrum of the cube $Q_{2m-1}$ is given by
\begin{equation}
\Sp_{\NLp}(Q_{2m-1}) = \left\{\frac{2k}{2m-1} : k=0,\ldots,2m-1\right\}.
\end{equation}
Let $t_{G} = (2m-1)\pi$. 
By Fact \ref{fact:normalized-p3}, the spectrum of $P_{3}$ is given by
$\Sp_{\NLp}(P_{3}) = \{0,1,2\}$ and it has perfect state transfer at time $t_{G}$
relative to the normalized Laplacian.
Note that
\begin{equation}
(2m-1)\pi \times \left\{\frac{2k}{(2m-1)} : k=0,\ldots,2m-1\right\} \times \{0,1,2\} \subseteq 2\pi\ZZ.
\end{equation}
Thus, by Theorem \ref{thm:normalized-weak-pst1}, 
$P_{3} \times Q_{2m-1}$ has perfect state transfer at time $t_{G}$ relative to the normalized Laplacian.
\end{proof}

\begin{figure}[t]
\begin{center}
\begin{tikzpicture}[scale=1.95]

\foreach \x in {60,90,120,...,390}
    \node at (\x:1)[circle,fill=black][scale=0.5] {};

\foreach \x in {60,150,240,330}
    \foreach \y in {120,210,300}
    {
        \draw (\x:1)--(\x+\y:1);
    }
\foreach \x in {120,210,300,390}
    \foreach \y in {60,150,240}
    {
        \draw (\x:1)--(\x+\y:1);
    }
\foreach \x in {90,180,270,360} {
    \foreach \y in {60,150,240}
    {
        \draw (\x:1)--(\x+\y:1);
    }
    \foreach \y in {120,210,300}
    {
        \draw (\x:1)--(\x+\y:1);
    }
}
\node at (60:1)[circle,fill=white][scale=0.5] {};
\draw[line width=0.2mm] (60:1) circle (0.055cm);
\node at (120:1)[circle,fill=white][scale=0.5] {};
\draw[line width=0.2mm] (120:1) circle (0.055cm);
\end{tikzpicture}
\quad \quad \quad \quad \quad
\begin{tikzpicture}
%
\node at (0,0)[scale=0.5]{};

\node at (0,0)[circle, fill=black][scale=0.5]{};
\node at (0,+3)[circle, fill=white][scale=0.5]{};
	\draw[line width=0.2mm] (0,+3) circle (0.1cm);
\node at (-1,+2)[circle, fill=black][scale=0.5]{};
\node at (+1,+1)[circle, fill=black][scale=0.5]{};
\draw (0,0)--(-1,+2);
\draw (0,0)--(+1,+1);
\draw (-1,+2)--(0,+3);
\draw (+1,+1)--(0,+3);
\draw (-1,+2)--(+1,+1);

\foreach \z in {+1.85}
{
	\foreach \x in {0,\z}
	{
	\draw (0+\x,0)--(-1+\x,+2);
	\draw (0+\x,0)--(+1+\x,+1);
	\draw (-1+\x,+2)--(0+\x,+3);
	\draw (+1+\x,+1)--(0+\x,+3);
	\draw (-1+\x,+2)--(+1+\x,+1);

	\node at (0+\x,0)[circle, fill=black][scale=0.5]{};
	\node at (0+\x,+3)[circle, fill=black][scale=0.5]{};
	\node at (-1+\x,+2)[circle, fill=black][scale=0.5]{};
	\node at (+1+\x,+1)[circle, fill=black][scale=0.5]{};
	}
	\foreach \x in {\z}
	{
	\draw (0,0)--(\x,0);
	\draw (+1,+1)--(+1+\x,+1);
	\draw (-1,+2)--(-1+\x,+2);
	\draw (0,+3)--(0+\x,+3);
	}

	\node at (0,+3)[circle, fill=white][scale=0.5]{};
	\draw[line width=0.2mm] (0,+3) circle (0.1cm);
	\node at (\z,0)[circle, fill=white][scale=0.5]{};
	\draw[line width=0.2mm] (\z,0) circle (0.1cm);
}
\end{tikzpicture}
\caption{Some graph products with perfect state transfer (between vertices marked white):
(a) the weak product $P_{3} \times K_{4}$ has perfect state transfer at time $3\pi$ 
relative to the normalized Laplacian.
(b) the Cartesian product $K_{2} \cart (\comp{K_{2}} + K_{2})$ has perfect state transfer at time $\pi/2$
relative to the standard Laplacian.
}
\label{fig:product}
\end{center}
\hrule
\end{figure}
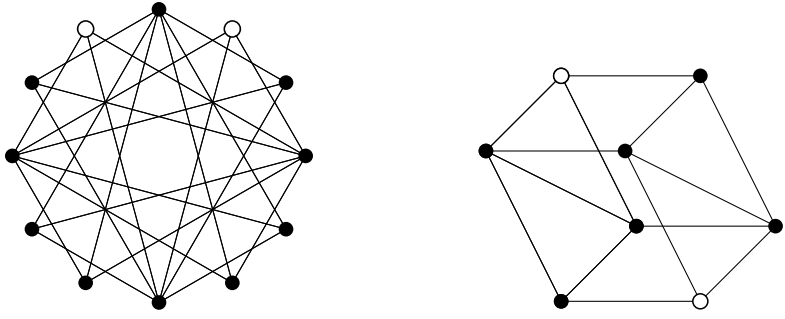

\smallskip

We show another closure property for perfect state transfer under weak products relative to
the normalized Laplacian.

\begin{theorem} \label{thm:normalized-weak-pst2}
Let $G$ and $H$ be graphs with perfect state transfer between vertices $g_{1},g_{2}$ 
and $h_{1},h_{2}$, respectively, both at time $\pst$ relative to the normalized Laplacian. 
Suppose that 
\begin{equation}
\pst \Sp_{\NLp}(G)\Sp_{\NLp}(H) \subseteq 2\pi\ZZ.
\end{equation} 
Then, $G \times H$ has perfect state transfer between vertices $(g_{1},h_{1})$ and $(g_{2},h_{2})$
at time $\pst$ relative to the normalized Laplacian.
\end{theorem}
\begin{proof}
By Lemma \ref{lemma:normalized-qwalk}, we have
\begin{equation}
\tbracket{\buket{(g_{2},h_{2})}}{e^{-it\NLp(G \times H)}}{\buket{(g_{1},h_{1})}}
=
\sum_{k} e^{-it\lambda_{k}} \tbracket{\buket{g_{2}}}{E_{k}}{\buket{g_{1}}} 
	\sum_{\ell} e^{it\lambda_{k}\mu_{\ell}} e^{-it\mu_{\ell}} \tbracket{\buket{h_{2}}}{F_{\ell}}{\buket{h_{1}}}.
\end{equation}
where $\NLp(G) = \sum_{k} \lambda_{k} E_{k}$ and $\NLp(H) = \sum_{\ell} \mu_{\ell} F_{\ell}$
are the spectral decompositions of the normalized Laplacians of $G$ and $H$.
Since $\pst\Sp_{\NLp}(G)\Sp_{\NLp}(H) \subseteq 2\pi\ZZ$, we have
\begin{eqnarray}
\tbracket{\buket{(g_{2},h_{2})}}{e^{-i\pst\NLp(G \times H)}}{\buket{(g_{1},h_{1})}}
	& = & \sum_{k} e^{-i\pst\lambda_{k}} \tbracket{\buket{g_{2}}}{E_{k}}{\buket{g_{1}}} 
			\tbracket{\buket{h_{2}}}{e^{-i\pst\NLp(H)}}{\buket{h_{1}}} \\
	& = & \tbracket{\buket{g_{2}}}{e^{-i\pst\NLp(G)}}{\buket{g_{1}}} 
			\tbracket{\buket{h_{2}}}{e^{-i\pst\NLp(H)}}{\buket{h_{1}}},
\end{eqnarray}
which proves the claim.
\end{proof}

\par\noindent{\em Remark}: Examples of of graphs which are realizations of 
Theorem \ref{thm:normalized-weak-pst2} have proved elusive.


\subsection{Paths}

We show that paths of length at least four have no antipodal perfect state transfer relative to
the normalized Laplacian. This nearly matches the situation in the adjacency matrix model
(see Christandl \etal \cite{cdel04,cddekl05} and Godsil \cite{godsil-dm11}).
We show a connection between paths under the normalized Laplacian and even cycles under the adjacency matrix. 
This connection seems well-known (see Aldous and Fill \cite{aldous-fill}), but we state a version 
useful for quantum walks.

\medskip

\begin{lemma} \label{lemma:normalized-path-cycle}
Let $n \ge 2$ be an integer.
The path $P_{n}$ has antipodal perfect state transfer relative to the normalized Laplacian
if and only if 
the cycle $C_{2(n-1)}$ has antipodal perfect state transfer relative to the adjacency matrix.
\end{lemma}
\begin{proof}
Let $m = n-1$.
Consider the cycle $C_{2m}$ with the vertex set $\{0,1,\ldots,2m-1\}$ 
where vertex $j$ is adjacent to vertex $k$ whenever $j-k \equiv \pm 1\pmod{2m}$.
Let $\pi$ be an equitable partition of $C_{2m}$ with $m+1$ cells where 
$V_{0} = \{0\}$, $V_{m} = \{m\}$, and $V_{k} = \{k, 2m-k\}$, for $k=1,\ldots,m-1$.
Then, $C_{2m}/\pi$ is a weighted path $\tilde{P}_{m+1}$ with adjacency matrix $A(\tilde{P}_{m+1})$
defined as:
\begin{equation}
\tbracket{\buket{j}}{A(\tilde{P}_{m+1})}{\buket{k}}
=
(\sqrt{2})^{\iverson{\beta(j,k)}} \tbracket{\buket{j}}{A(P_{m+1})}{\buket{k}} 
\end{equation}
where $\beta(j,k)$ holds if either $j$ or $k$ is a boundary vertex in $\{0,m\}$.
\ignore{
That is,
\begin{equation}
A(\tilde{P}_{m+1})
=
\begin{bmatrix}
0 & \sqrt{2} & 0 & \ldots & 0 & 0 & 0 \\
\sqrt{2} & 0 & 1 & \ldots & 0 & 0 & 0 \\
0 & 1 & 0 & \ldots & 0 & 0 & 0 \\
\vdots & \vdots & \vdots & \ldots & \vdots & \vdots & \vdots \\
0 & 0 & 0 & \ldots & 0 & 1 & 0 \\
0 & 0 & 0 & \ldots & 1 & 0 & \sqrt{2} \\
0 & 0 & 0 & \ldots & 0 & \sqrt{2} & 0 
\end{bmatrix}
\end{equation}
}
We note that 
\begin{equation}
A(\tilde{P}_{m+1}) = 2 D^{-1/2}A(P_{m+1})D^{-1/2},
\end{equation}
where $D$ is the degree matrix of $P_{m+1}$.
\ignore{
For brevity, let $D = D(P_{m+1})$ and $A = A(P_{m+1})$.
To see this, note that
\begin{equation}
2 \bra{j} D^{-1/2}AD^{-1/2}\ket{k}
\textstyle
	= 2 \frac{1}{\sqrt{\deg(j)\deg(k)}} \bra{j}A\ket{k} \\
	= (\sqrt{2})^{\iverson{\beta(j,k)}} \bra{j}A\ket{k}.
\end{equation}
}
Therefore, we have
\begin{equation}
\textstyle \NLp(P_{m+1}) = \Id - \frac{1}{2}A(C_{2m}/\pi).
\end{equation}
This shows that
\begin{equation}
\tbracket{\buket{m}}{e^{-it\NLp(P_{m+1})}}{\buket{0}}
	= e^{-it} \tbracket{\buket{V_{m}}}{e^{itA(C_{2m}/\pi)}}{\buket{V_{0}}} 
	= e^{-it} \tbracket{\buket{m}}{e^{itA(C_{2m})}}{\buket{0}},
\end{equation}
where the last equality follows by {\em lifting} (see Lemma \ref{lemma:lifting}).
\end{proof}

We will need the following results for our main theorem in this section.

\begin{proposition} \label{prop:godsil-vertex-transitive}
(Godsil, Corollary 8.2.2. in \cite{godsil-book}) \\
If perfect state transfer occurs on a connected vertex-transitive graph $G$, then
the eigenvalues of $G$ are integers.
\end{proposition}

\begin{fact} \label{fact:olmstead} (Olmstead, see Corollary 3.12 in Niven \cite{niven}) \\
If $\theta \in 2\pi\QQ$, then the only rational values of
$\cos(\theta)$ are $0, \pm\frac{1}{2}, \pm 1$.
\end{fact}

\begin{theorem} \label{thm:no-normalized-pst-path}
For $n \ge 4$, there is no antipodal perfect state transfer on $P_{n}$ relative to the normalized Laplacian.
\end{theorem}
\begin{proof}
Let $n \ge 4$.
By Lemma \ref{lemma:normalized-path-cycle}, if $P_{n}$ has antipodal perfect state transfer 
relative to the normalized Laplacian, then $C_{2(n-1)}$ has antipodal perfect state transfer 
relative to the adjacency matrix.
By Proposition \ref{prop:godsil-vertex-transitive}, since $C_{2(n-1)}$ is connected and vertex-transitive, 
if it has perfect state transfer, then its eigenvalues must be integers.
But the eigenvalues of the cycles $C_{2(n-1)}$ are given by
$\{2\cos(2\pi k/2(n-1)) : 0 \le k < 2(n-1)\}$. 
By Fact \ref{fact:olmstead}, these are integers only at $0,\pm 1$ which implies that
$k/(n-1) \in \ZZ/2$. This implies that $n = 2,3$, which is a contradiction.
\end{proof}

\medskip

\par\noindent{\em Remark}:
The application of Fact \ref{fact:olmstead} in the proof of Theorem \ref{thm:no-normalized-pst-path}
followed similar ideas used by Godsil in the context of standard Laplacians (see \cite{godsil-book}).


\section{Conclusions}

In this work, we studied perfect state transfer in quantum walk relative to graph Laplacians. 
As pointed out by Bose, Casaccino, Mancini and Severini \cite{bcms09},
a quantum walk relative to the standard Laplacian is related to quantum spin networks 
in the isotropic Heisenberg ($\xyz$ interaction) model whereas a quantum walk with the 
adjacency matrix is connected to the $\xy$ model.
In their seminal work, Farhi and Gutmann \cite{fg98} used a weighted Laplacian matrix to 
define {\em continuous-time quantum walks} to underscore the close connection with continuous-time 
random walks. In the first work which introduced {\em perfect state transfer}, Bose \cite{bose03} studied
quantum spin chains in the $\xyz$ (or ``Laplacian'') model.

Our main goal in this work is to understand perfect state transfer in quantum walks relative to 
the standard, signless and normalized Laplacians.
Our focus was on irregular graphs (since all Laplacian quantum walks are equivalent otherwise)
and nonbipartite graphs (since the standard/signless Laplacian quantum walks are equivalent otherwise).
To the best of our knowledge, the signless and normalized Laplacians have not been studied extensively 
in the context of quantum walks.
Although the signless Laplacian of a graph $G$ has no clear ``physical'' motivation, 
it shares a strong spectral bond with the line graph $\Line(G)$.
So, it provides a method for analyzing quantum walk on line graphs in the $\xy$ model. 
This is a direction which merits closer study. 
In contrast, the normalized Laplacian has a clear ``physical'' meaning, albeit in a more classical sense.
It has been closely studied in connection with the Heat Kernel random walk in spectral graph theory
(see Chung \cite{chung}) and in machine learning (see Kondor and Lafferty \cite{kl02}).

We observed a useful closure property relative to the standard Laplacian: 
{\em complementation} preserves perfect state transfer. 
Relative to the adjacency matrix (and perhaps the other two Laplacians), 
this property holds only for regular graphs.
This closure property allowed us to characterize Laplacian perfect state transfer on double cones.
In turn, we generalized a known result of Bose \etal \cite{bcms09} and found a much simpler proof.
We also found families of double cones with perfect state transfer relative to the
signless Laplacian, but not relative to the standard Laplacian.
Our proofs relied on ideas from the theory of equitable and almost-equitable partitions.

By exploiting the connection between signless Laplacians and line graphs, 
we showed some negative results for perfect state transfer relative to the signless Laplacian.
Using a reduction to the adjacency matrix model, we observed that paths with five or more vertices 
have no antipodal perfect state transfer relative to the signless Laplacian 
(also standard, by switching equivalence).
But, a better negative result is known for paths (due to Godsil \cite{godsil-book}) and, 
recently, for trees (due to Coutinho and Liu \cite{cl14}).
We applied our techniques to nonbipartite graphs and showed this for the simplest family 
of odd unicyclic graphs (two pendant paths attached to a three cycle). 
Our proof made heavy use of Godsil's results on controllable subsets of graphs \cite{godsil-ac12}.

A paradoxical lore of quantum walk on the $n$-cube relative to the adjacency matrix is 
that its (antipodal) perfect state transfer time is $\pi/2$, for {\em any} $n$. This striking
statement seems to violate the constant speed of light postulate.
In contrast, relative to the normalized Laplacian, the $n$-cube has antipodal perfect state transfer 
at time $n\pi/2$.
This example suggests that quantum walks relative to the normalized Laplacian might be closer to reality.
Here, we proved another closure property for perfect state transfer but under the {\em weak product}.
This is not too surprising given that the normalized Laplacian spectrum behaves well under weak product 
(and not under, say, Cartesian product). 
As a corollary, we showed that a weak product of $P_{3}$ with either an even clique or odd cube has 
perfect state transfer.  It is curious that $P_{3}$ has perfect state transfer under the normalized Laplacian 
but not relative to the standard/signless Laplacians. To complete the picture, we showed that paths 
with four or more vertices do not have (antipodal) perfect state transfer under the normalized Laplacian. 
This almost matches the state of affairs under the adjacency matrix, where no perfect state transfer exists
between {\em any} pair of vertices (see Godsil \cite{godsil-dm11}). It is unclear if this stronger result 
holds relative to the normalized Laplacian.


\section{Acknowledgments}

We would like to thank Ada Chan, Gabriel Coutinho, and Chris Godsil for their generous and helpful comments.
The research of the first five authors was supported by NSF grant DMS-1262737 and NSA grant H98230-14-1-0141.
The research of H.Z. is supported by a Graduate Student Fellowship at the University of Waterloo
while working the guidance of Chris Godsil.




\end{document}